\documentclass[11pt]{article}
\usepackage{graphicx} 
\usepackage{geometry}[1inch]
\usepackage[utf8]{inputenc}
\usepackage{mathrsfs}
\usepackage{amssymb,latexsym,amsmath}
\usepackage{amsthm}
\usepackage{color}
\usepackage{xspace} 
\usepackage{tcolorbox}
\usepackage{hyperref}
\usepackage{thmtools}
\usepackage{cleveref}
\usepackage{wrapfig}
\usepackage{thm-restate}
\usepackage{subcaption}
\usepackage{algorithmic}

\usepackage{algorithm}
\usepackage{float}
\usepackage{enumitem}
\usepackage{booktabs}

\newtheorem{theorem}{Theorem}
\newtheorem{lemma}{Lemma}[section]

\theoremstyle{definition}

\crefname{lemma}{Lemma}{Lemmas}
\crefname{appendix}{Appendix}{Appendices}
\crefname{proposition}{Proposition}{Propositions}
\crefname{observation}{Observation}{Observations}
\crefname{claim}{Claim}{Claims}
\crefname{figure}{Figure}{Figures}

\DeclareMathOperator{\bc}{bc}
\DeclareMathOperator{\tin}{in}
\DeclareMathOperator{\tout}{out}
\newcommand{\SPV}{\text{SP}_V}
\newcommand{\SPE}{\text{SP}_E}

\newcommand{\OPT}{\operatorname{OPT}}
\newcommand{\ALG}{\operatorname{ALG}}

\definecolor{DeepGreen}{RGB}{0, 150, 70}

\newcommand{\Zx}[1]{\{0,1,\dots,#1-1\}}

\title{Optimal Non-Adaptive Vantage Point Selection}
\author{Jie Gao\thanks{\texttt{jg1555@cs.rutgers.edu}. Gao would like to
acknowledge NSF support through IIS-2229876, DMS-2220271, CNS-2515159, DMS-2311064, and CCF-2118953.}\\Rutgers University \and Nicole Wein\thanks{\texttt{nswein@umich.edu}. Supported by NSF CAREER award 2541910. }\\University of Michigan \and Chang Wu\thanks{\texttt{wu-c24@mails.tsinghua.edu.cn}.}\\Tsinghua University}
\date{}

\begin{document}

\maketitle
\begin{abstract}

We study the \emph{vantage point selection} problem, introduced by Ashvinkumar, Chowdhury, Gao, Goswami, Mitchell, and Polishchuk [WADS'25] to model the problem of estimating bottleneck capacities on the Internet.  The input is a weighted undirected graph with unique shortest paths where every edge has a distinct unknown \emph{capacity}. When the algorithm \emph{queries} a vertex $v$, it reveals the minimum-capacity edge on the shortest path from $v$ to every other vertex reachable from $v$. The goal is to maximize the total number of revealed edges. The quality of an algorithm is measured by its competitive ratio against an optimal algorithm that knows all edge capacities a priori. 

We first consider the foundational single-query setting, where both the algorithm and the optimal algorithm are restricted to a single query. There is a trivial upper bound of $O(n)$ on the competitive ratio and the best known lower bound was $\tilde{\Omega}(\sqrt{n})$. We provide an algorithm and matching lower bound (up to polylogarithmic factors) showing that the best possible competitive ratio is $\tilde{\Theta}(n^{2/3})$. 

Furthermore, we extend our results to the general setting where the optimal algorithm is allowed $k$ queries and our algorithm is allowed $\alpha k$ queries for $\alpha\geq 1$. We present a randomized non-adaptive algorithm and matching lower bound (up to polylogarithmic factors) showing that the best possible expected competitive ratio for non-adaptive algorithms is the following surprisingly complex bound:
$$
\tilde{\Theta}\left(
\min\left\{
\frac{n}{\alpha k},
\max\left(
\sqrt{\frac{n}{\alpha}},
\frac{n^{2/3}}{\alpha k^{1/3}}
\right)
\right\}
\right).
$$
\end{abstract}

\pagenumbering{gobble}
\clearpage

\pagebreak
\pagenumbering{arabic}

\section{Introduction}

We study the \emph{vantage point selection} problem, introduced by Ashvinkumar, Chowdhury, Gao, Goswami, Mitchell, and Polishchuk~\cite{Ashvinkumar2025-bh}. The input is a weighted undirected graph where the shortest path between each pair of vertices (if it exists) is unique. Additionally, every edge has a distinct \emph{capacity} that is unknown to the algorithm. The algorithm can perform \emph{queries} on vertices, where \texttt{query($v$)} reveals the minimum-capacity edge on the shortest path from $v$ to every other vertex reachable from $v$. The goal of the algorithm is to maximize the total number of revealed edges. The quality of an algorithm is measured by its \emph{competitive ratio} against an optimal algorithm that knows all edge capacities a priori. We consider the general version of the problem where the optimal algorithm is allowed $k$ queries, and our algorithm is allowed $\alpha k$ queries for some $\alpha\geq 1$. That is, for $\alpha=1$ both algorithms get the same number of queries, and for $\alpha>1$, our algorithm is permitted some \emph{resource augmentation} in the number of queries. We study the \emph{non-adaptive} setting, where our algorithm cannot use the information gained from previous queries to inform future queries. 

\paragraph{Motivation} Inferring network topology and parameters from end-to-end measurements has been an active research topic in network tomography and Internet measurement~\cite{Vardi1996-cy, 10.1214/088342304000000422,Coates2002-xq}. 
In a network, the \emph{capacity} of an edge is the maximum bit rate that can be delivered along it. The maximum data rate along a path is determined by the \emph{minimum} edge capacity, also known as the \emph{bottleneck} edge. Estimating edge capacities and bottleneck edges is important for many applications, especially those that are real-time, throughput-sensitive (e.g., video conferencing, multiplayer gaming, and VR/AR applications)~\cite{Allman1999-ex, Guerrero2010-bx, Banerjee2000-cn, Harfoush2003-mw}. Existing methods rely on the central idea of measuring round-trip time through probing messages~\cite{Salcedo2018-kx,Huang2020-uy,Prasad2003-sv}. 
For example, a \textsc{TraceRoute} probing message 
issued from a source $s$ to a destination $t$ returns, for \emph{each} intermediate node along the shortest path, the round-trip time (RTT). Therefore, the RTT at a vertex right after a bottleneck edge shows a big jump compared to the previous vertex. This can be used to identify the bottleneck edge on the path and estimate its capacity. However, issuing such probing messages from a vertex/node in the network
 is not easy and requires gaining access to the system at that node, which has non-trivial logistical costs~\cite{Jueckstock2019-ym}. Thus, multiple prior works have discussed how to choose vantage points -- the vertices where probing messages are issued~\cite{Bottger2017-oc,Burger2014-hb,Holterbach2017-ec}.

This scenario motivated the authors in
\cite{Ashvinkumar2025-bh}  to define the vantage point selection problem. In particular, querying a vertex $v$ corresponds to issuing a probing message from $v$ to all other vertices to identify the corresponding bottleneck edges. In this view, the vantage point selection problem captures the goal of issuing probing messages from only a limited number of vertices to identify the maximum number of bottleneck edges. The non-adaptive version is particularly relevant, as the algorithm must depend solely on the structure of the graph, and is oblivious to edge capacities and throughput, which can be dynamic due to real-time traffic. 

In addition to the practical motivation, this problem is theoretically interesting in its own right. We believe that our techniques contribute to a broader understanding of how to extract structure from systems of shortest paths that exhibit varying levels of congestion, expansion, and connectivity.

\paragraph{Prior Work}$ $\\
\textbf{Single-query setting.} 
For simplicity, we begin with the setting where both our algorithm and the optimal algorithm get a single query ($k=\alpha=1$). We use $\OPT$ and $\ALG$ to denote the number of edges revealed by the optimal algorithm and our algorithm, respectively. We use the convention of expressing the competitive ratio as a value $\geq 1$, so we define the competitive ratio as $\OPT/\ALG$. 

As a baseline, achieving competitive ratio $n-1$ for an $n$-vertex graph is trivial for the following reason. Any query can only reveal at most $n-1$ edges (one edge per shortest path from the queried vertex to the other $n-1$ vertices), and the algorithm can always guarantee revealing at least one edge if the graph is not empty (querying a vertex $v$ reveals each of $v$'s incident edges that constitute the shortest path to one of $v$'s neighbors). For deterministic algorithms, there is an asymptotically matching lower bound of $\Omega(n)$, where the rough intuition is that hard instances become easier to construct when one can place the minimum-capacity edges incident to the algorithm's deterministic queries~\cite{Ashvinkumar2025-bh}. Thus, the interesting regime of the problem is for randomized algorithms. Before the current paper, no non-trivial randomized algorithm was known. 

However, algorithms were known for some special classes of graphs: competitive ratio $O(\sqrt{n})$ was known for trees, and $O(n^{2/3})$ for planar graphs~\cite{Ashvinkumar2025-bh}. 

From the lower bounds side, prior work~\cite{Ashvinkumar2025-bh} provided a simple lower bound of $\tilde{\Omega}(\sqrt{n})$, that works even for trees ($\tilde{\Omega}$ hides polylogarithmic factors). The construction is simply $\sqrt{n}$ disjoint paths of length $\sqrt{n}$ each. The capacities on one randomly chosen path $P^*$ are decreasing along the path, while the capacities on the edges of the rest of the paths are chosen randomly. (Note that this is sufficient for specifying the capacities since only their relative values matter, not their absolute values.) The optimal algorithm knows the capacities and can thus choose to query the first vertex $v$ of $P^*$. This reveals all of the edges on $P^*$ since the minimum-capacity edge of the shortest path from $v$ to any vertex on $P^*$ is the last edge on that path. Thus, $\OPT=\sqrt{n}$. On the other hand, since the algorithm does not know the capacities, it views all of the paths as indistinguishable. So, the idea is that the algorithm cannot do better than picking a path at random and querying its first vertex (without loss of generality up to constant factors). With probability $1/\sqrt{n}$ the algorithm chooses $P^*$ and reveals $\sqrt{n}$ edges. Otherwise, the algorithm chooses a path with randomly chosen capacities. The number of revealed edges is the number of \emph{record-breaking} capacities along that path (i.e., capacities that are smaller than all previous capacities along the path), which is $O(\log n)$ in expectation. Thus, the algorithm reveals $O(\log n)$ edges in expectation, yielding a competitive ratio of $\tilde{\Omega}(\sqrt{n})$. 

In summary, for the single-query setting, there is a large gap between the known bounds, both of which are very simple to prove, leaving the question:

\begin{center}
\emph{Question: Can we get any non-trivial (sublinear competitive ratio) single-query algorithm for vantage point selection? How close to the lower bound of $\tilde{\Omega}(\sqrt{n})$ can we get?}
\end{center}

\noindent \textbf{General multi-query setting.} Now we consider the general setting where the optimal algorithm gets $k$ queries and our algorithm gets $\alpha k$ queries for $\alpha\geq 1$. When $\alpha=n/k$, the algorithm gets $n$ queries, so the competitive ratio is trivially 1. Beyond such trivial types of analysis, no upper bounds were known for general graphs. 

However, it was known that the single-query algorithms for trees and planar graphs extend to the general multi-query setting.  In particular, for trees the competitive ratio is $O\big(\frac{\sqrt{n}}{\alpha\sqrt{k}}\big)$, and for planar graphs it is $O\big(\frac{n^{2/3}}{\alpha k^{2/3}}\big)$~\cite{Ashvinkumar2025-bh}. 
Additionally, an extension of the single-query lower bound provides a general lower bound of $\tilde{\Omega}\big(\frac{\sqrt{n}}{\alpha\sqrt{k}}\big)$, which works even for trees and matches the upper bound for trees up to polylogarithmic factors~\cite{Ashvinkumar2025-bh}.

In summary, for the general multi-query setting there is a large gap between the known bounds, leaving the question:

\begin{center}
\emph{Question: Can we get any non-trivial algorithm for vantage point selection? How close to the lower bound of $\tilde{\Omega}\big(\frac{\sqrt{n}}{\alpha\sqrt{k}}\big)$ can we get?}
\end{center}


\noindent \textbf{Related settings.}  Prior work~\cite{Ashvinkumar2025-bh} studies several related settings. They consider the case where shortest paths can be non-unique, and ties can be broken differently for different shortest-path trees. In this case, they show a prohibitive competitive ratio lower bound of $\Omega(n^{1-\varepsilon})$ for any constant $\varepsilon$. 

They also consider the stochastic setting in which all capacities are chosen randomly, and the goal is to maximize the expected number of revealed edges. They show that this setting admits much better bounds than the adversarial setting, and obtain an algorithm with a constant competitive ratio for any $k$ with $\alpha=1$. 

Lastly, they consider the computational complexity of the problem of finding the optimal $k$ queries when the capacities are known. They provide a simple proof that the problem is NP-hard by reduction from vertex cover. 

\subsection{Our Results}


We answer both questions by obtaining the first algorithm with non-trivial competitive ratio for both the single-query and multi-query settings with and without resource augmentation. We also provide lower bounds for non-adaptive algorithms that are \emph{tight} up to polylogarithmic factors for the entire range of parameters. 

Although our result in the multi-query setting is a strict generalization of our single-query result, we first present our single-query result for the sake of simplicity:

\begin{restatable}{theorem} {single-main}
\label{thm:single-main}
Given a weighted $n$-vertex $m$-edge graph with unique shortest paths between reachable pairs, the best possible expected competitive ratio for the single-query ($k=\alpha=1$) vantage point selection problem is
$$\frac{\OPT}{\mathbb{E}(\ALG)}=\tilde{\Theta}(n^{2/3}).$$
There is a randomized algorithm achieving this competitive ratio that runs in time $\tilde{O}(n(n+m))$.
\end{restatable}

Recall that for the single-query setting it was previously known that the correct answer is somewhere between $\tilde{\Omega}(\sqrt{n})$ and $O(n)$; \Cref{thm:single-main} shows that it is in fact $\tilde{\Theta}(n^{2/3})$.

For the general multi-query setting we show that the correct answer is surprisingly complex, with three different cases for different parameter regimes.  We use $\OPT_k$ to denote the number of edges revealed by an optimal algorithm with $k$ queries, and $\ALG_{\alpha k}$ to denote the number of edges revealed by our algorithm with $\alpha k$ queries. Our main theorem is as follows.

\begin{restatable}{theorem} {main}
\label{thm:main}
Given a weighted $n$-vertex $m$-edge graph with unique shortest paths between reachable pairs, for any positive integers $\alpha,k$ with $\alpha k\le n$, the best possible expected competitive ratio for the vantage point selection problem in the non-adaptive setting is
$$
\frac{\OPT_k}{\mathbb{E}(\ALG_{\alpha k})}=
\tilde{\Theta}\!\left(
\min\left\{
\frac{n}{\alpha k},
\max\left(
\sqrt{\frac{n}{\alpha}},
\frac{n^{2/3}}{\alpha k^{1/3}}
\right)
\right\}
\right).
$$ There is a randomized non-adaptive algorithm achieving this competitive ratio that runs in time $\tilde{O}(n(n+m))$.
\end{restatable}


This is both the first non-trivial algorithm for this problem on general graphs, and it is tight (up to polylogarithmic factors) for the entire range of parameters.

We remark that the $\tilde{O}(n(n+m))$ running time of our algorithm means that the algorithm provides no additional overhead (up to polylogarithmic factors) beyond computing all-pairs shortest paths. 

To help interpret the rather complicated bound of \Cref{thm:main}, we provide two plots in \Cref{fig:improvement-bounds}. The right plot visualizes the entire parameter regime in three dimensions, while the left plot visualizes the case where the optimal algorithm gets a single query ($k=1$).

\begin{figure}[H]
    \centering
    \includegraphics[width=1\textwidth]{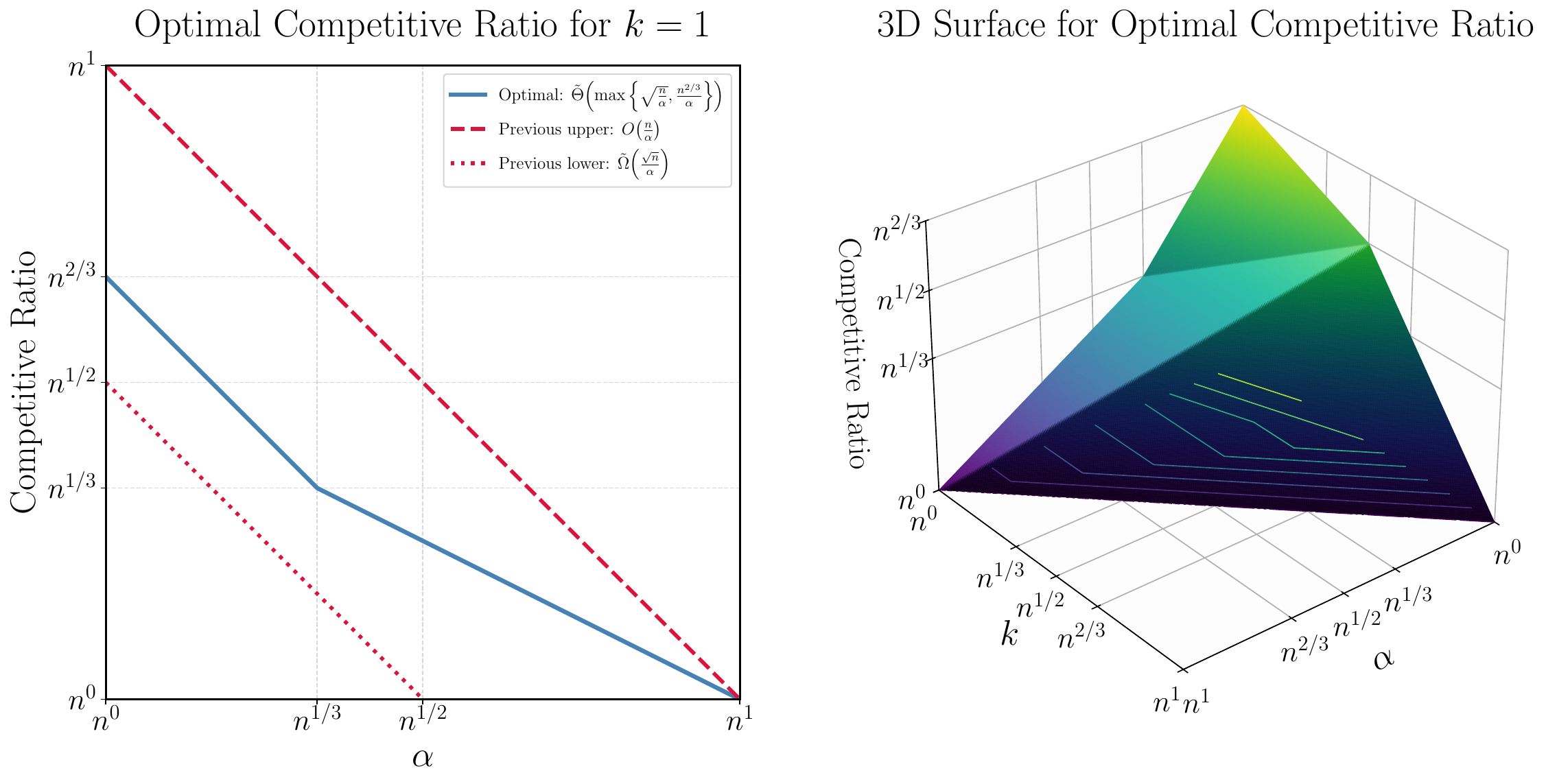}
    \caption{Optimal expected competitive ratio for non-adaptive vantage point selection, with polylogarithmic factors suppressed. \textbf{Left:} The case $k=1$, showing our tight bound together with the previous upper and lower bounds. The $y$-intercept at $n^{2/3}$ corresponds to the single-query setting of \Cref{thm:single-main}, while the $x$-intercept corresponds to the trivial regime in which the algorithm gets $n$ queries and the competitive ratio is $1$. \textbf{Right:} The three-dimensional surface over the $(\alpha,k)$ parameter space for the tight bound.} 
      \label{fig:improvement-bounds}
\end{figure}


We leave as an open problem whether our lower bounds can be extended to the adaptive setting, or if there is an adaptive algorithm that surpasses our non-adaptive lower bound.

\section{Technical Overview}\label{sec:tech}

\subsection{Algorithm}

For simplicity, in this overview we will focus on the single-query ($\alpha=k=1$) setting. Recall that in this setting, our algorithm achieves competitive ratio $\tilde{O}(n^{2/3})$. We will start by outlining a simpler algorithm with competitive ratio $\tilde{O}(n^{6/7})$ for the class of graphs with hop-diameter $\tilde{O}(1)$. After that, we will describe some ideas for improving the algorithm to $\tilde{O}(n^{2/3})$ for general graphs.

\paragraph{$\tilde{O}(n^{6/7})$-competitive algorithm for low hop-diameter graphs.}
Let $v_{\OPT}$ be the vertex queried by the optimal algorithm. 
A natural idea is to query a vertex incident to an edge with high \emph{betweenness centrality}, i.e.~an edge on many shortest paths. Intuitively, one may hope that such a vertex hits many of the shortest paths from $v_{\OPT}$. Indeed, this is one ingredient of the algorithm: With probability 1/3, we uniformly sample a vertex from the set $S$ of vertices incident to an edge with high betweenness centrality. Here, ``high'' means at least $n^{8/7}$. It will be useful to note that $|S| =\tilde{O}(n^{6/7})$, which is true for the following reason. The hop-diameter is $\tilde{O}(1)$, so the total number of pairs of the form (shortest path $p$, edge on $p$) in the entire graph is $\tilde{O}(n^2)$, and each edge of high betweenness centrality contributes $\geq n^{8/7}$ to this total.

It turns out that this type of query only provides competitive ratio  $\tilde{O}(n^{6/7})$ when the graph structure and capacity assignment have certain properties. This is why we perform this query with only probability 1/3, and later we will establish the other types of queries needed. First, we will describe the scenario in which the above query is effective. The important property is: 

\begin{center} $\leq n^{6/7}$ of vertices are such that their shortest path from $v_{\OPT}$ avoids $S$. 
\end{center}

In other words, we say that the \emph{residual degree} of a vertex $v$ with respect to $S$ is the number of reachable vertices $u$ for which the shortest $uv$-path avoids $S$. Using this terminology, the above property says that $v_{\OPT}$ has residual degree $\leq n^{6/7}$.

If this property holds, we are done for the following reason. First, since the shortest path from $v_{\OPT}$ to a fixed vertex only reveals one  bottleneck edge, the shortest paths from $v_{\OPT}$ that avoid $S$ can only reveal $\leq n^{6/7}$ edges. Furthermore, if $\OPT<2\cdot n^{6/7}$ then we are done because any algorithm that reveals at least one edge is already $O(n^{6/7})$-competitive. Combining these two observations, at least half of the edges revealed by $v_{\OPT}$ are on shortest paths from $v_{\OPT}$ that are hit by $S$. For each edge $e$ on such a shortest path $p$ hit by $S$, querying the vertex in $S$ on $p$ also reveals $e$, since $e$ is the bottleneck edge for the entire path $p$. Thus, if we queried all $\tilde{O}(n^{6/7})$ vertices in $S$, we would reveal a constant fraction of the edges revealed by $v_{\OPT}$. So, if we pick one randomly, $\mathbb{E}[\ALG] = \tilde{\Omega}(\OPT/n^{6/7})$ as desired.

Now we need to consider the case when the residual degree of $v_{\OPT}$ with respect to $S$ is more than $n^{6/7}$. We will condition on the number of other vertices in the graph that also have high ($>n^{6/7}$) residual degree with respect to $S$. If this number is at most $n^{6/7}$, then we claim we are done by simply randomly sampling such a vertex. This is because we would query $v_{\OPT}$ itself with probability $\Omega(1/n^{6/7})$, so $\mathbb{E}[\ALG] = \Omega(\OPT/n^{6/7})$ as desired. That is, the second step of the algorithm is to sample, with 1/3 probability, from the set of vertices with high residual degree with respect to $S$.

The last remaining case is when ``many'' ($>n^{6/7}$) vertices in the graph have ``high'' ($>n^{6/7}$) residual degree with respect to $S$. Let $P$ be the set of ordered pairs $(u,v)$ of vertices where $u$ has a high residual degree, and the shortest $uv$-path avoids $S$. There are $>n^{12/7}$ such pairs. By the definition of $S$, every edge on a shortest path between a pair of vertices in $P$ has betweenness centrality at most $n^{8/7}$. We will use these observations to argue that, for the remaining 1/3 probability, it suffices for the algorithm to take a uniform random sample from all vertices in the graph. Our goal is to show that the total sum over all vertices $v$ of the number of edges $v$ reveals is $\tilde{\Omega}(n^{8/7})$, which implies that a uniformly random vertex reveals $\tilde{\Omega}(n^{1/7})$ edges in expectation, yielding competitive ratio $\tilde{O}(n^{6/7})$.

For each edge $e$ on some path in $P$, let $Q_e$ denote the set of vertex pairs in $P$ whose bottleneck edge is $e$. The betweenness centrality constraint on such edges means that $|Q_e|\leq n^{8/7}$. For all $e$, any query vertex that is an endpoint of a pair in $Q_e$ reveals $e$, and the number of such endpoint vertices is at least $\sqrt{|Q_e|}$. Thus, $\sum_e \sqrt{|Q_e|}$ is a lower bound on the quantity we wish to bound: the total sum over all vertices $v$ of the number of edges $v$ reveals. Since there are $>n^{12/7}$ pairs in $P$, and each pair belongs to exactly one set $Q_e$, we have $\sum_e |Q_e|>n^{12/7}$. To minimize $\sum_e \sqrt{|Q_e|}$ subject to the constraints $|Q_e|\leq n^{8/7}$ and $\sum_e |Q_e|>n^{12/7}$, take $n^{4/7}$ of the $|Q_e|$ values to be $n^{8/7}$. For these $n^{4/7}$ values, $\sqrt{|Q_e|} = n^{4/7}$, so $\sum_e \sqrt{|Q_e|} \geq n^{4/7}\cdot n^{4/7} =n^{8/7}$, as desired.

In summary, the algorithm samples uniformly from the following three sets with probability 1/3 each:
\begin{itemize}[itemsep=0em]
    \item the set $S$ of vertices incident to an edge with high betweenness centrality,
    \item the set of vertices with high residual degree with respect to $S$,
    \item the set of all vertices in the graph.
\end{itemize}

\paragraph{Improving to $\tilde{O}(n^{2/3})$ for general graphs.}
The framework is similar to the previous algorithm: we would like to identify a ``small'' set $S$ of vertices so that if we consider the set $P$ of ordered pairs $(u, v)$ of reachable vertices whose shortest paths avoid $S$ and $u$ has high
residual degree, at least one of two objectives is satisfied:
\begin{enumerate}
    \item $v_{\OPT}$ is in ``few'' pairs in $P$ (i.e.~$v_{\OPT}$ has ``low'' residual degree with respect to $S$), OR
    \item all edges in the graph have ``low'' betweenness centrality with respect to $P$. 
\end{enumerate}
Objective 1 is analogous to the initial ``important property'' of the previous algorithm, and objective 2 is analogous to the last case of the previous algorithm which handles the pairs $P$ containing only edges of low residual betweenness centrality. 

In the previous algorithm, we had fixed thresholds for ``low'' residual degree and ``low'' betweenness centrality. The key development in the improved algorithm is to replace these fixed thresholds with parameters that are iteratively adjusted over the course of the algorithm, as the set $S$ is gradually built. This adds intricacy to the algorithm because the iteratively adjusted parameters are interdependent; changing one parameter changes the others. To achieve the right balance, the final algorithm has several nested subroutines that each adjust different parameters.

We begin by considering the subroutine for building $S$. Since $S$ should be built with objective 2 in mind, we would intuitively like to add to $S$ vertices of high betweenness centrality with respect to the current $P$ so that we are ultimately left with all edges only having low betweenness centrality with respect to $P$. The algorithm does this in a greedy manner: roughly speaking, it iteratively adds to $S$ the vertex of highest betweenness centrality with respect to the current $P$. 

This greedy subroutine is executed in phases, and after a phase of adding a subset of vertices to $S$, the algorithm checks if the maximum residual degree in the graph has decreased by a lot; this is our proxy for the residual degree of $v_{\OPT}$ in objective 1, since we cannot calculate that directly without knowing $v_{\OPT}$. (The final algorithm actually checks a slightly different quantity, but this is the intuition.) If the maximum residual degree has decreased by a lot, then we are intuitively in good shape with respect to objective 1. In this case, we proceed to the next phase of adding vertices to $S$, and if the maximum residual degree continues to decrease by a lot at each phase, we can bound the number of phases until $v_{\OPT}$ is only in a constant number of pairs in $P$. 

On the other hand, if the maximum residual degree does not decrease by a lot, then we claim that we are intuitively in good shape with respect to objective 2. This is because the maximum residual degree \emph{not} decreasing much is intuitively aligned with $|P|$ not decreasing much, which means that the vertices with highest betweenness centrality with respect to $P$ that we greedily added to $S$ actually had low residual betweenness centrality with respect to $P$, in line with objective 2.

A key challenge in designing the algorithm is the interdependence of the parameters: adding more vertices to $S$ decreases three quantities: the size of $P$, the maximum residual degree, and the betweenness centrality with respect to $P$. In the above analysis of the original $\tilde{O}(n^{6/7})$ algorithm, the calculations get more favorable as $|P|$ \emph{increases}, as the maximum residual degree \emph{decreases}, and as the betweenness centrality with respect to $P$ \emph{decreases}. 

Due to this interdependence, the final algorithm has multiple nested iterative procedures. In the outer loop, we run the entire algorithm with geometrically decreasing thresholds for the maximum residual degree. Within that, we perform the previously described procedure, which iteratively adds vertices to $S$ and checks if we have achieved objective 1 or 2. Finally, within that we perform the iterative procedure to greedily add vertices to $S$ one-by-one.

These are the main ideas behind the single-query version of the algorithm. For the multi-query version, we need a generalization of the algorithm in which the parameters $k$ and $\alpha$ are incorporated into the above framework and balanced with the other parameters.

Implementation: It turns out there is a straightforward implementation of our algorithm that runs in time $\tilde{O}(n^3)$. To reduce this to $\tilde{O}(n(n+m))$, we require additional data structures. 

\subsection{Lower Bound}


Since the full construction is rather involved, we present two simpler constructions that capture the main ideas behind the lower bound. The first construction is for the regime of large $\alpha$, $k$, while the second is for the single-query case $\alpha=k=1$. The full construction can be viewed as combining these two ideas, and its details are presented in \Cref{sec:lb}.

\subsubsection{Case 1: Large-parameter regime}


\begin{figure}[h!]
    \centering
    \includegraphics[width=0.6\textwidth]{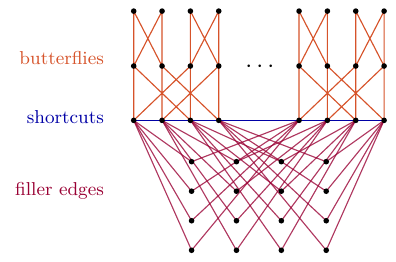}
    \caption{The lower-bound construction for the large-parameter regime.}
    \label{fig:case1}
\end{figure}


In this section, we will provide the ideas for a lower bound of 
$$\tilde{\Omega}\left(\min\left\{\frac{n}{\alpha k},\sqrt{\frac{n}{\alpha}}\right\}\right).$$
This reaches the optimal competitive ratio in the regime $\alpha^3k^2\ge n$.
The construction underlying this lower bound is illustrated in \Cref{fig:case1}.

We begin with the main intuition. First, it is useful to build large parts of the graph that are highly symmetric, in the sense that many vertices are indistinguishable from the algorithm's perspective. This allows us to hide the optimal vantage points among many candidate vertices: the hidden capacities can be chosen to favor a small subset of these candidates, while a non-adaptive algorithm cannot reliably distinguish this subset from the others.

Second, the optimal vantage points need to reveal many edges. A natural way to achieve this is to make the shortest-path trees rooted at the optimal vantage points exhibit expansion, i.e. to make them look more like balanced binary trees than long paths.

\emph{Butterfly graphs} are useful for this purpose: they provide both symmetry and expansion. Thus, the starting point of the construction is a collection of butterfly graphs, each of which contains one hidden optimal vantage point in its top layer.

This alone is not sufficient, because querying a non-optimal vertex in a butterfly may also reveal many edges. To overcome this issue, we further exploit the fact that each edge can be revealed only once. The goal is to make it so that the edges revealed by the $k$ optimal queries are almost disjoint, whereas the edges revealed by the $\alpha k$ non-optimal queries overlap heavily.

We realize this idea by introducing \emph{filler edges} and \emph{shortcuts} incident to the bottom of the butterflies. The filler edges are the many nearly unique edges revealed by the optimal queries, and their structure can be viewed as a collection of stars whose leaves are bottom butterfly vertices. The shortcuts form low-weight shortcut structures near the bottom of the butterflies, and their purpose is to isolate the filler edges from most non-optimal queries, preventing such queries from revealing too many filler edges.

By tuning the butterfly sizes and the number of filler edges, this construction gives the desired lower bound in the large-parameter regime.

Lastly, we note that this lower bound crucially only works in the non-adaptive setting. In the adaptive setting, querying a non-optimal butterfly vertex could reveal information about the location of the optimal query within that butterfly.

\subsubsection{Case 2: $\alpha=k=1$}
\begin{figure}[h!]
    \centering
    \includegraphics[width=0.8\textwidth]{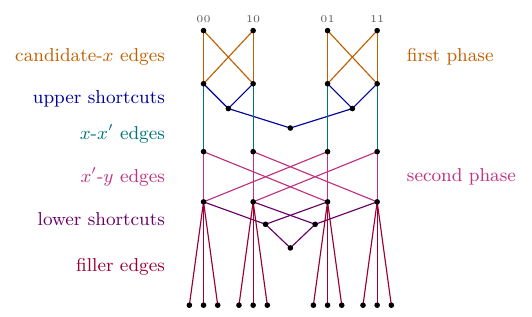}
    \caption{The lower bound construction for the single-query case with $w=2$. }
    \label{fig:case2}
\end{figure}

In the single-query case, our goal is to prove an $\tilde{\Omega}(n^{2/3})$ lower bound. The construction underlying this lower bound is illustrated in \Cref{fig:case2}.

Here, we need to hide the optimal vantage point among $\tilde{\Omega}(n^{2/3})$ candidate points, while still ensuring that querying the optimal vantage point reveals nearly linearly many edges. At the same time, a capacity-unaware query should reveal only $\tilde{O}(n^{1/3})$ edges.

The main issue is that the butterfly-based construction above no longer gives the right balance. Suppose we continue using butterflies, each of size $s$. If $s\le \sqrt n$, then there are not enough candidate points to hide the optimum well enough. On the other hand, if $s>\sqrt n$, then a random query inside the butterfly may reveal too many edges. We therefore replace the butterfly with a two-phase expansion.

 Fix an integer parameter $w=\tilde{\Theta}(n^{1/3})$. Let $u_{p,q}$, where $p,q\in \Zx{w}$, be the candidate optimal vantage points. These vertices form the candidate layer. We also introduce three layers with the same index set: an $x$-layer, an $x'$-layer, and a $y$-layer. Finally, as in the previous construction, we add a filler layer that supplies the many edges revealed by the hidden optimal query. For the moment, we focus on the non-filler part of the construction.

In the \emph{first phase}, each $u_{p,q}$ is connected to $x_{p',q}$ for every $p'\in\Zx{w}$, forming the \emph{candidate-$x$ edges}. Between the $x$-layer and the $x'$-layer, the \emph{$x$-$x'$ edges} are one-to-one: $x_{p,q}$ is connected to $x'_{p,q}$. In the \emph{second phase}, each $x'_{p,q}$ is connected to $y_{p,q'}$ for every $q'\in\Zx{w}$, forming the \emph{$x'$-$y$ edges}. Consequently, after two phases, every candidate optimal point can reach every vertex in the $y$-layer. From the perspective of the shortest-path trees rooted at the candidate points, we have replaced balanced binary trees with perfect $O(n^{1/3})$-ary trees of height two.

We must still prevent arbitrary queries in this structure from revealing too many edges. To do this, we introduce shortcut structures on the $x$-layer and the $y$-layer, referred to as \emph{upper shortcuts} and \emph{lower shortcuts}, respectively. Each shortcut structure is itself a perfect $O(n^{1/3})$-ary tree of height two. The point is that, for a fixed query vertex, most shortest paths that use a shortcut pass through the top level of such a shortcut tree, which contains only $\tilde{O}(n^{1/3})$ edges. We set the capacities of these top-level shortcut edges to be small, so that they become the bottleneck edges on these paths. This bounds the number of distinct edges that such a query can reveal through the shortcut structures.

The remaining important paths are those from the candidate layer to the $y$-layer, which are arranged not to use the shortcut structures. We set the capacities of the $x$-$x'$ edges to be smaller than those of the other layer-to-layer edges on these paths, so that the contribution of all such candidate-to-$y$ paths is also bounded. The filler layer is then attached beyond the $y$-layer so that the hidden optimal candidate can still reveal nearly linearly many distinct filler edges, while non-optimal queries remain limited by the shortcut and two-phase-expansion structure described above.

\section{Preliminaries}

Let $G=(V,E)$ be an undirected weighted graph with $|V|=n$. We assume throughout this paper that the shortest path between every pair of vertices (if it exists) is unique. For vertices $u,v\in V$, write $u\sim v$ if $u$ and $v$ lie in the same connected component. Let 
$$
\mathcal P := \{(u,v)\in V\times V \mid u\ne v \text{ and } u\sim v\}
$$
be the set of reachable ordered pairs.
For $(u,v)\in \mathcal P$, let $\SPV(u,v)$ and $\SPE(u,v)$ denote, respectively, the set of vertices (including $u$ and $v$) and the set of edges on the unique shortest path from $u$ to $v$. 

Each edge $e\in E$ is associated with a capacity value $c(e)>0$. The capacities are fixed but unknown to the algorithm. We further assume that all capacities are distinct. Hence, for every ordered pair $(u,v)\in \mathcal P$, the shortest path from $u$ to $v$ contains a unique edge with minimum capacity. We denote this edge by
$$
b(u,v):=\arg\min_{e\in \SPE(u,v)} c(e),
$$
and refer to it as the \emph{bottleneck edge} of the pair $(u,v)$.

\paragraph{Query model.}
A \emph{query} consists of selecting a vertex $u\in V$ as a \emph{vantage point}. Querying $u$ reveals the bottleneck edge on the shortest path from $u$ to every other vertex in the same connected component. Formally, define
$$
R(u):=\{b(u,v)\mid (u,v)\in\mathcal P\}\subseteq E.
$$
Thus, $R(u)$ is the set of edges revealed by querying $u$.

If the algorithm outputs an ordered list of vantage points $v_1,\dots,v_{\alpha k}$, then the total set of revealed edges is
$$
\bigcup_{i=1}^{\alpha k} R(v_i).
$$

\paragraph{Objective.}
The algorithm is given the graph structure, edge weights, and hence the shortest-path system of $G$, but it does not know the capacities $\{c(e)\}_{e\in E}$. Its goal is to choose a small number of vantage points so as to reveal as many bottleneck edges as possible.

For a positive integer $k$, define $\OPT_k$ as the maximum number of edges that can be revealed by choosing $k$ vantage points, assuming full knowledge of the hidden capacities, i.e.,
$$
\OPT_k
:=
\max_{s_1,\dots,s_k\in V}
\left|
\bigcup_{i=1}^k R(s_i)
\right|.
$$

If an algorithm outputs $\alpha k$ vantage points $v_1,\dots,v_{\alpha k}$, we write
$$
\ALG_{\alpha k}
:=
\left|
\bigcup_{i=1}^{\alpha k} R(v_i)
\right|
$$
for the number of edges revealed by the algorithm. When the algorithm is randomized, $\ALG_{\alpha k}$ is a random variable, and guarantees are understood in expectation unless stated otherwise. 

\paragraph{Competitiveness with an augmented query budget.}
We compare an algorithm that is allowed to make $\alpha k$ queries against the benchmark optimum that uses only $k$ queries, where $\alpha\ge 1$ is an augmentation factor. This bicriteria viewpoint is the one studied throughout the paper: the algorithm is given a larger (or equal) query budget, and must compete with the optimal capacity-aware choice of only $k$ vantage points. Note that our results also cover the non-augmented version where $\alpha=1$. 

\paragraph{Non-adaptivity.}
The focus of this paper is the \emph{non-adaptive} setting. A non-adaptive algorithm must choose all of its vantage points before seeing any capacity information or any revealed bottleneck edges. Equivalently, its output distribution may depend on the graph and the shortest-path structure, but not on the hidden capacity assignment.

\section{Upper Bound}

In this section we prove the upper bound part of \Cref{thm:main}, which we restate here:

\begin{theorem}
    Given a weighted $n$-vertex $m$-edge graph with unique shortest paths between reachable pairs, for any positive integers $\alpha,k$ with $\alpha k\le n$, there is a randomized non-adaptive algorithm with running time $\tilde{O}(n(n+m))$ and competitive ratio
$$
\frac{\OPT_k}{\mathbb{E}(\ALG_{\alpha k})}=
\tilde{O}\!\left(
\min\left\{
\frac{n}{\alpha k},
\max\left(
\sqrt{\frac{n}{\alpha}},
\frac{n^{2/3}}{\alpha k^{1/3}}
\right)
\right\}
\right).
$$ 
\end{theorem}


The algorithm is built from 4 subroutines. We will first outline how each subroutine corresponds to the following description of the structure of the algorithm from the technical overview (\Cref{sec:tech}):
``The final algorithm has multiple nested iterative procedures. In the outer loop, we run the entire algorithm with geometrically decreasing thresholds for the maximum residual degree. Within that, we perform the previously described procedure, which iteratively adds vertices to $S$ and checks if we have achieved objective 1 or 2. Finally, within that we perform an iterative procedure to greedily add vertices to $S$ one-by-one.''

\Cref{alg:final} is the shell of the final algorithm, which sets the parameters, calls \Cref{alg:IR}, and makes the final selection of vantage points to query. \Cref{alg:IR} is the ``outer loop'' from the previous paragraph, which iteratively calls  \Cref{alg:ROR}. \Cref{alg:ROR} is the procedure that iteratively adds vertices to $S$ and checks if we have achieved objective 1 or 2: to add to $S$, it calls \Cref{alg:high-bc}, and if it ever achieves objective 2, it calls \Cref{alg:RWRBC} to select the appropriate vantage points.

\subsection{The Algorithm}

In this subsection, we present the algorithm and the main lemmas proved in the analysis. We first introduce the notation used throughout the remainder of this section.

Recall that $\mathcal{P}$ denotes the set of all reachable ordered pairs, i.e.,
$$
\mathcal P := \{(u,v)\in V\times V \mid u\ne v \text{ and } u\sim v\}
$$
For any $S\subseteq V$, let $P(S)$ be the set of pairs in $\mathcal{P}$ whose shortest paths are hit by $S$: 
$$
P(S):=\{(s,t)\in \mathcal{P}\mid \SPV(s,t)\cap S\neq\emptyset\}.
$$
Thus, for any pair family $P\subseteq \mathcal{P}$, the \emph{residual pairs} not hit by $S$ are $P\setminus P(S)$. 

For any $P\subseteq \mathcal{P}$ and $v\in V$, we denote by $\deg_P(v)$ the number of pairs in $P$ whose first endpoint is $v$, called the \emph{degree} of $v$ with respect to $P$:
$$
\deg_P(v):=\left|\{u\in V\mid (v,u)\in P\}\right|.
$$
In particular, for any $S\subseteq V$, define the \emph{residual degree} of $v$ with respect to $S$ by
$$
\deg_S(v):=\deg_{\mathcal{P}\setminus P(S)}(v).
$$ 

For any $P\subseteq V\times V$ and $e\in E$, let $\bc_P(e)$
denote the \emph{betweenness centrality} of an edge $e$, i.e., the number of shortest paths induced by pairs in $P$ that traverse $e$:
$$
\bc_P(e):=\left|\{(s,t)\in P\mid e\in \SPE(s,t)\}\right|.
$$
Similarly, we also define the betweenness centrality of a vertex $v$ with respect to shortest paths of pairs in $P$:
$$
\bc_P(v):=\left|\{(s,t)\in P\mid v\in \SPV(s,t)\}\right|.
$$
For any $S\subseteq V$, we write
$$
\bc_S(e):=\bc_{\mathcal{P}\setminus P(S)}(e)
\qquad\text{and}\qquad
\bc_S(v):=\bc_{\mathcal{P}\setminus P(S)}(v)
$$
for the \emph{residual betweenness centrality} of an edge $e$ and a vertex $v$ with respect to $S$.

Finally, let
$$
\bc(P):=\max_{e\in E}\bc_P(e)
$$
denote the maximum \emph{edge} betweenness centrality with respect to $P$. For any $S\subseteq V$, we also write
$$
\bc(S):=\bc(\mathcal{P}\setminus P(S))
$$
for the maximum residual edge betweenness centrality with respect to $S$. Clearly, $\bc(P)$ is at most the maximum \emph{vertex} betweenness centrality with respect to $P$.

Although the graph is undirected, we work with ordered pairs for notational convenience, as the ordering will be useful when specifying which vertices to query.

Throughout the algorithms below, the graph $G=(V,E)$ and the parameters $\alpha$ and $k$ are fixed globally. For notational simplicity, we do not list them explicitly as inputs to the subroutines.

\subsubsection{Finding Vertices of High Residual Betweenness Centrality}

We begin with the subroutine in \Cref{alg:high-bc}. Given a set $P$ of ordered pairs and an integer $I$, the goal is to find a set $S \subseteq V$ with $|S| \le I$ such that the maximum residual edge betweenness centrality is not greater than $|P|/I$. Specifically, we use a greedy process: we initialize $S$ to the empty set, and at each iteration we add to $S$ the vertex $v$ with the highest betweenness centrality among shortest paths that avoid existing vertices in $S$. This procedure stops when we have $I$ vertices in $S$. We also return the value $B$, defined as the maximum residual edge betweenness centrality with respect to $P \setminus P(S)$.

\begin{algorithm}[H]
\caption{\textsc{High-Residual-BC}$(P,I)$: selecting vertices of high residual betweenness centrality.}
\label{alg:high-bc}
\begin{algorithmic}[1]
\REQUIRE A set $P\subseteq \mathcal P$ of ordered pairs; a positive integer $I$.
\ENSURE A set $S\subseteq V$ with $|S|\le \min\{I,n\}$ and a value $B$. 
\STATE $S\gets\emptyset$
\STATE $P_{cur}\gets P$ 
\FOR{$i=1$ to $\min\{I,n\}$}
    \FOR{each vertex $v\in V\setminus S$}
        \STATE Compute
        $$
        \bc_{P_{cur}}(v)
        =
        \left|
        \left\{
        (s,t)\in P_{cur}\mid v\in \SPV(s,t)
        \right\}
        \right|.
        $$
    \ENDFOR
    \STATE Let $v_i\in \arg\max_{v\in V\setminus S}\bc_{P_{cur}}(v)$ \label{line:HBC6}
    \STATE $S\gets S\cup\{v_i\}$
    \STATE $P_{cur}\gets P_{cur}\setminus P(\{v_i\})$ 
\ENDFOR
\STATE $B\gets \bc(P_{cur})$
\STATE \textbf{return} $(S,B)$
\end{algorithmic}
\end{algorithm}

We restate the main lemma of \Cref{alg:high-bc} for reference from the analysis.
\begin{restatable}{lemma} {lemone}
\label{lem:high-bc-threshold}
Let $(S,B)=\textsc{High-Residual-BC}(P,I)$. Then
$$
B=\bc(P\setminus P(S))\le \frac{|P|}{I}.
$$
\end{restatable}

\subsubsection{Querying Subroutine for Low Betweenness Centrality Regime}

Suppose now that we are given a set $T\subseteq V$ and a set $P\subseteq T\times V$ of ordered pairs such that the maximum edge betweenness centrality with respect to $P$ is at most $B$. The querying procedure in \Cref{alg:RWRBC} is designed to handle this case: for each query vertex, with probability $1/2$ we query a vertex sampled uniformly at random from $T$, and otherwise we query a vertex sampled uniformly at random from $V$. 

\begin{algorithm}[H]
\caption{\textsc{RevealWithRestrictedBC}$(T)$}
\label{alg:RWRBC}
\begin{algorithmic}[1]
\REQUIRE A set $T\subseteq V$.
\ENSURE $\alpha k$ vantage points $\{v_1,\dots,v_{\alpha k}\}$.
\FOR{$i=1$ to $\alpha k$}
    \STATE With probability $1/2$, sample $v_i\sim \mathrm{Unif}(T)$.
    \STATE Otherwise, sample $v_i\sim \mathrm{Unif}(V)$.
\ENDFOR
\STATE \textbf{return} $\{v_1,\dots,v_{\alpha k}\}$
\end{algorithmic}
\end{algorithm}

Using the intuition from the technical overview, when all edges have bounded betweenness centrality, random queries from $T$ and $V$ have a good chance to reveal edges that are competitive with the optimal choice. In the analysis, we will prove that \Cref{alg:RWRBC} gives the following lower bound on the expected number of revealed edges. 

\begin{restatable}{lemma} {lemtwo}
\label{lem:rwrbc}
Consider a set of ordered pairs $P\subseteq (T\times V)\cap \mathcal P$ with $\bc(P)\le B$. Let $\{v_1,\dots,v_{\alpha k}\}$ be the query vertices produced by $\textsc{RevealWithRestrictedBC}(T)$. Then
$$
\mathbb{E}\!\left[\left|\bigcup_{i=1}^{\alpha k} R(v_i)\right|\right]
\ge
\Omega\!\left(
\min\left\{
\frac{|P|}{B},
\alpha k\sqrt{\frac{|P|^2}{B|T|n}}
\right\}
\right).
$$
\end{restatable}

\subsubsection{Reducing the Maximum Residual Degree or Revealing Many Edges}
With the two modules we have introduced so far, one natural idea is to apply \Cref{alg:high-bc} to add vertices of high betweenness centrality to $S$, so that the remaining pairs $P\setminus P(S)$ have maximum edge betweenness centrality $B$, and then invoke \Cref{alg:RWRBC}. However, as suggested by the technical overview, controlling residual betweenness centrality alone is not enough: the guarantee of \Cref{alg:RWRBC} also depends on how many residual pairs remain and on the number of vertices from which these pairs originate. We therefore also keep track of residual degrees.

To make \Cref{alg:RWRBC} effective, we restrict attention to vertices of large residual degree. More precisely, for a threshold $d$, we let $T$ be the set of vertices whose degree with respect to $P\setminus P(S)$ is at least $d$, and we let $Q$ be the set of residual pairs in $P\setminus P(S)$ whose first endpoint lies in $T$. The subroutine in \Cref{alg:ROR} either drives all residual degrees below $d$, or finds a stage at which $Q$ still contains many pairs while the greedy step has already made the maximum edge betweenness centrality of $Q$ small. In the latter case, the low-betweenness querying subroutine \Cref{alg:RWRBC} can be applied to $T$ and yields a good revealing strategy.

One subtlety is that after applying \Cref{alg:high-bc} to gather more vertices of high betweenness centrality into $S$, the residual pair set $P\setminus P(S)$ changes, thereby altering both the high-degree vertex set $T$ and the pair family $Q$. For this reason, \Cref{alg:ROR} uses an iterative procedure. In each iteration, it forms the current high-degree set $T$ and the corresponding pair family $Q$, applies \Cref{alg:high-bc} to $Q$, and then updates $S$, $T$, and $Q$. If both $T$ and $Q$ remain large after the update, then the instance has not shrunk much, while the new residual betweenness centrality bound is low enough for \Cref{alg:RWRBC} to be effective. Otherwise, at least one of $T$ or $Q$ shrinks by a constant factor, so this can happen only logarithmically many times before no high-degree vertices remain.

\begin{algorithm}[H]
\caption{\textsc{ReduceOrReveal}$(P,d,I)$}
\label{alg:ROR}
\begin{algorithmic}[1]
\REQUIRE A set $P\subseteq \mathcal P$ of ordered pairs; a positive real number $d$; a positive integer $I$. 
\ENSURE A set $S\subseteq V$ with $|S|=O(I\log n)$ and $\alpha k$ vantage points $\{v_1,\dots,v_{\alpha k}\}$. 
\STATE $S\gets\emptyset$
\STATE $T\gets \{u\in V\mid \deg_{P}(u)\ge d\}$
\STATE $Q\gets \{(u,v)\in P\mid u\in T\}$
\WHILE{$T\neq\emptyset$}
    \STATE $(S',B)\gets \textsc{High-Residual-BC}(Q,I)$ \label{line:ROR-HBC}
    \STATE $S\gets S\cup S'$ \label{line:ROR-update-S}
    \STATE $T_{\mathrm{old}}\gets T,\qquad Q_{\mathrm{old}}\gets Q$ \label{line:ROR-old}
    \STATE $T\gets \{u\in V\mid \deg_{P\setminus P(S)}(u)\ge d\}$ \label{line:ROR-update-T}
    \STATE $Q\gets \{(u,v)\in P\setminus P(S)\mid u\in T\}$ \label{line:ROR-update-Q}
    \IF{$|T|>\frac{|T_{\mathrm{old}}|}{2}$ \textbf{and} $|Q|>\frac{|Q_{\mathrm{old}}|}{2}$} \label{line:ROR-large}
        \STATE $\{v_1,\dots,v_{\alpha k}\}\gets \textsc{RevealWithRestrictedBC}(T)$
        \STATE \textbf{return} $(S,\{v_1,\dots,v_{\alpha k}\})$
    \ENDIF
\ENDWHILE
\STATE Sample $v_i\sim \mathrm{Unif}(V)$ independently for all $i=1,\dots,\alpha k$. \COMMENT{dummy vertices} \label{line:ROR15}
\STATE \textbf{return} $(S,\{v_1,\dots,v_{\alpha k}\})$ 
\end{algorithmic}
\end{algorithm}

In the analysis, we will prove the following guarantee for \Cref{alg:ROR}: one of the following two conditions holds. Note that the first condition is satisfied when the algorithm reaches line~\ref{line:ROR15}, so the vertices sampled there are merely dummy vertices.

\begin{restatable}{lemma} {lemthree}
\label{lem:ror}
Let $(S,\{v_1,\dots,v_{\alpha k}\})=\textsc{ReduceOrReveal}(P,d,I)$. Then $|S|=O(I\log n)$, and one of the following two conclusions holds:
\begin{enumerate}
    \item For all $v\in V$, $\deg_{P\setminus P(S)}(v)<d$; or
    \item $\begin{aligned}[t]
    \mathbb{E}\!\left[\left|\bigcup_{i=1}^{\alpha k} R(v_i)\right|\right]
    \ge\Omega\!\left(\min\left\{I,\alpha k\sqrt{\frac{dI}{n}}\right\}\right).
    \end{aligned}$
\end{enumerate}
\end{restatable}

Note that after calling this subroutine, we can determine which conclusion holds without knowing the capacities, simply by checking whether $\max_{v\in V}\deg_{P\setminus P(S)}(v)<d$.

\subsubsection{Iterative Residual Degree Reduction}
We now iterate the previous subroutine over geometrically decreasing residual degree thresholds. Specifically, we maintain an upper bound $D$ on the maximum residual degree in the current pair family, and define the next threshold $d$ as a function of $D$ and a parameter $\beta$. The parameter $\beta$ will later be chosen to match the target competitive ratio, and it controls the tradeoff between how quickly the residual degree bound decreases and how strong the revealing guarantee is.

After applying \Cref{alg:ROR}, either the maximum residual degree drops below $d$, in which case we update the degree upper bound from $D$ to $d$ and proceed to the next iteration, or the procedure already reveals many edges and we terminate. If $d$ is chosen to be too large, then the reduction from $D$ to $d$ is too mild; if $d$ is too small, then the revealing guarantee becomes too weak. We therefore choose $d$ carefully to balance these two outcomes.

\begin{algorithm}[H]
\caption{\textsc{IterativeReduce}$(\beta,I)$}
\label{alg:IR}
\begin{algorithmic}[1]
\REQUIRE A real number $\beta>0$; a positive integer $I$.
\ENSURE A set $S\subseteq V$ with $|S|=O(I\log^2 n)$ and $\alpha k$ vantage points $\{v_1,\dots,v_{\alpha k}\}$.
\STATE $S\gets\emptyset$
\STATE $D\gets n$ \label{line:IR-init-D}
\STATE $\mathcal{P}\gets \{(u,v)\in V\times V\mid u\neq v\text{ and }u\sim v\}$
\WHILE{$D\ge 1$}
    \STATE $d\gets \dfrac{k^2D^2n}{\beta^2(\alpha k)^2I}$ \label{line:IR-d}
    \STATE $(S',\{v'_1,\dots,v'_{\alpha k}\})\gets \textsc{ReduceOrReveal}(\mathcal P\setminus P(S),d,I)$ \label{line:IR-ROR}
    \STATE $S\gets S\cup S'$ \label{line:IR-update-S}
    \IF{$\max_{v\in V}\deg_{\mathcal P\setminus P(S)}(v)\ge d$} \label{line:IR-check}
        \STATE \textbf{return} $(S,\{v'_1,\dots,v'_{\alpha k}\})$
    \ENDIF
    \STATE $D\gets d$ \label{line:IR-update-D}
\ENDWHILE
\STATE Sample $v_i\sim \mathrm{Unif}(V)$ independently for all $i=1,\dots,\alpha k$. \COMMENT{dummy vertices} \label{line:IR12}
\STATE \textbf{return} $(S,\{v_1,\dots,v_{\alpha k}\})$
\end{algorithmic}
\end{algorithm}



In the analysis, we will prove the following guarantee for \Cref{alg:IR}.

\begin{restatable}{lemma} {lemfour}
\label{lem:ir}
Suppose that
$$
\frac{k^2n^2}{\beta^2(\alpha k)^2I}\le \frac{1}{2}
\qquad\text{and}\qquad
I\ge \frac{kn}{\beta}.
$$
Let $(S,\{v_1,\dots,v_{\alpha k}\})=\textsc{IterativeReduce}(\beta,I)$ and $\hat{d}=\max_{v\in V}\deg_{\mathcal P\setminus P(S)}(v)$. Then $|S|=O(I\log^2 n)$, and 
$$
\mathbb{E}\!\left[\left|\bigcup_{i=1}^{\alpha k} R(v_i)\right|\right]
\ge
\tilde{\Omega}\!\left(\frac{k\hat{d}}{\beta}\right).
$$
\end{restatable}

\subsubsection{The Final Algorithm}

We now combine the three cases handled by the previous components. The first case is when many edges in the optimum can already be revealed by vertices in the hitting set $S$. The second case is when the remaining residual instance is handled by \Cref{lem:ir}. The third case is when sampling uniformly from all of $V$ is already sufficient because the query budget $\alpha k$ is large.

Before stating the algorithm, we choose a parameter $\beta$ that will later correspond to the competitive ratio. In the analysis, we will verify that this choice of $\beta$ satisfies the conditions required by \Cref{lem:ir}.

\begin{algorithm}[H]
\caption{A randomized non-adaptive
$\tilde{O}\!\left(\min\left(\dfrac{n}{\alpha k},\max\left(\sqrt{\dfrac{n}{\alpha}},\dfrac{n^{2/3}}{\alpha k^{1/3}}\right)\right)\right)$-competitive algorithm.}
\label{alg:final}
\begin{algorithmic}[1]
\REQUIRE The globally fixed graph $G=(V,E)$ with unique shortest paths and globally fixed parameters $\alpha,k$ with $\alpha k\le n$.
\ENSURE $\alpha k$ vantage points $\{v_1,\dots,v_{\alpha k}\}$. 
\STATE Choose $\beta$ of order
$$
\tilde{\Theta}\!\left(
\max\left\{
\sqrt{\frac{n}{\alpha}},
\frac{n^{2/3}}{\alpha k^{1/3}}
\right\}
\right)
$$
large enough that, after setting $I\gets \lceil \alpha k\beta\rceil$, we have
$$
\frac{k^2n^2}{\beta^2(\alpha k)^2I}\le \frac{1}{2}
\qquad\text{and}\qquad
I\ge \frac{kn}{\beta}.
$$
\STATE $(S,\{v'_1,\dots,v'_{\alpha k}\})\gets \textsc{IterativeReduce}(\beta,I)$
\STATE Draw $X\in\{1,2,3\}$ uniformly at random.
\IF{$X=1$}\label{line:final-branch1}
    \IF{$S\neq\emptyset$}
        \STATE Sample $v_i\sim \mathrm{Unif}(S)$ independently for all $i=1,\dots,\alpha k$.
    \ELSE
        \STATE Sample $v_i\sim \mathrm{Unif}(V)$ independently for all $i=1,\dots,\alpha k$.  
    \ENDIF
\ELSIF{$X=2$}\label{line:final-branch2}
    \STATE Set $v_i\gets v'_i$ for all $i=1,\dots,\alpha k$.
\ELSE\label{line:final-branch3}
    \STATE Sample $v_i\sim \mathrm{Unif}(V)$ independently for all $i=1,\dots,\alpha k$.
\ENDIF
\STATE \textbf{return} $\{v_1,\dots,v_{\alpha k}\}$
\end{algorithmic}
\end{algorithm}

We will show that \Cref{alg:final} achieves the competitive ratio claimed in \Cref{thm:main}.

\subsection{Competitive Ratio Analysis}


\subsubsection{Proof of \Cref{lem:high-bc-threshold}}
\lemone*

\begin{proof}
If $I>n$, then \Cref{alg:high-bc} selects all vertices, so $S=V$ and hence $P\setminus P(S)=\emptyset$. Therefore, $B=0\le |P|/I$, and the claim is trivially true. We may thus assume that $I\le n$.

For $i=0,1,\dots,I$, let $S_i$ be the set maintained by \textsc{High-Residual-BC} after $i$ iterations. Thus $S_0=\emptyset$, $S_i=S_{i-1}\cup \{v_i\}$ for $i\ge 1$, and $S_I=S$. 
Also let
$$
P_i:=P\setminus P(S_i),
$$
and for $i=1,\dots,I$ define
$$
B_i:=\max_{v\in V\setminus S_{i-1}}\bc_{P_{i-1}}(v)
=\bc_{P_{i-1}}(v_i).
$$

Since $P_i\subseteq P_{i-1}$ for every $i$, the sequence $B_1,B_2,\dots,B_I$ is non-increasing. Moreover, when we add $v_i$ to $S$, the pairs in $P_{i-1}$ whose shortest path contains $v_i$ do not appear in $P_i$. Hence
$$
|P_{i-1}\setminus P_i|\ge B_i.
$$
Summing over all iterations yields
$$
\sum_{i=1}^I B_i
\le
\sum_{i=1}^I |P_{i-1}\setminus P_i|
\le
|P|.
$$
Because the sequence is non-increasing, we obtain
$$
B_I\le \frac{1}{I}\sum_{i=1}^I B_i\le\frac{|P|}{I}.
$$

Now take an edge $e \in E$ on the shortest path of a pair $(s,t)\in P_I$. Let $x$ be an endpoint of $e$. Then $x$ also lies on $\SPV(s,t)$. Since $(s,t)\in P_I$, we have $\SPV(s,t)\cap S_I=\emptyset$, and therefore $x\notin S_I$. It follows that
$$
\bc_{P_I}(e)\le \bc_{P_I}(x)\le \max_{v\in V\setminus S_I}\bc_{P_I}(v)\le B_I.
$$
Taking the maximum over all edges gives
$$
B
=
\bc(P\setminus P(S))
=
\max_{e\in E}\bc_{P_I}(e)
\le
B_I
\le
\frac{|P|}{I}.
$$
\end{proof}

\subsubsection{Proof of \Cref{lem:rwrbc}}

\lemtwo*

\begin{proof}
Since $\bc(P)\le |P|$, we may assume without loss of generality that $B\le |P|$.
Replacing $B$ by $\min\{B,|P|\}$ only strengthens the claimed lower bound.

Now consider the fixed edge capacities. 
For each edge $e\in E$, let $Q_e$ be the set of pairs in $P$ whose bottleneck edge is $e$:
$$
Q_e:=\{(s,t)\in P\mid b(s,t)=e\}.
$$
Since every pair in $Q_e$ has $e$ on its shortest path and $e$ lies on at most $\bc_P(e)\le B$ shortest paths induced by pairs in $P$, we have
$$
|Q_e|\le B.
$$

Define
$$
A_e:=\{s\in T\mid \exists t\in V\text{ such that }(s,t)\in Q_e\},
\qquad
a_e:=|A_e|,
$$
$$
B_e:=\{t\in V\mid \exists s\in T\text{ such that }(s,t)\in Q_e\},
\qquad
b_e:=|B_e|.
$$
Clearly, $Q_e\subseteq A_e\times B_e$, and we have $|Q_e|\le a_eb_e$. Combining this with $|Q_e|\le B$, we get
$$
|Q_e|
\le
\min\{a_eb_e,B\}
\le
\sqrt{a_eb_e\cdot B}.
$$
Therefore,
$$
|P|
=
\sum_e |Q_e|
\le
\sqrt{B}\sum_e \sqrt{a_eb_e}
\le
\sqrt{B\left(\sum_e a_e\right)\left(\sum_e b_e\right)},
$$
where the last step applies the Cauchy--Schwarz inequality.
Rearranging gives
$$
\left(\sum_e a_e\right)\left(\sum_e b_e\right)\ge \frac{|P|^2}{B}.
$$
Define $R:=\frac{|P|}{\sqrt{B|T|n}}$. Hence either
$$
\sum_e a_e\ge |T|R
\qquad\text{or}\qquad
\sum_e b_e\ge nR.
$$

Now consider one query $v$ generated by \Cref{alg:RWRBC}. For a fixed edge $e$, the edge $e$ is revealed by a query at $v$ if the query $v$ either lands in $A_e$ when sampling uniformly from $T$ or lands in $B_e$ when sampling uniformly from $V$. Thus
$$
\Pr[e\in R(v)]
\ge
\frac{a_e}{2|T|}+\frac{b_e}{2n}.
$$
Summing over all edges and applying the inequalities above, we have   
$$
\sum_e \Pr[e\in R(v)]
\ge
\frac{1}{2|T|}\sum_e a_e+\frac{1}{2n}\sum_e b_e
\ge
\frac{R}{2}
=
\frac{|P|}{2\sqrt{B|T|n}}.
$$

The analysis so far shows that one query node using \Cref{alg:RWRBC} reveals a good number of edges in expectation. Next, we consider multiple query nodes. Specifically, we
apply the same idea not only to the original pair set $P$, but also to the residual pairs whose bottleneck edges have not yet been revealed.

For $j=0,1,\dots,\alpha k$, let $F_j$ be the set of edges revealed so far by query vertices $\{v_1,\dots,v_j\}$. 
$$
F_0=\emptyset, \qquad F_j:=\bigcup_{i=1}^j R(v_i),
\qquad
U_j:=|F_j|.
$$
Define $P_j$ as the set of residual pairs after removing the pairs whose bottleneck edge has been revealed in $F_j$:
$$
P_j:=\{(s,t)\in P\mid b(s,t)\notin F_j\}.
$$
Conditioned on the first $j-1$ query vertices, we can apply the same one-query analysis to $P_{j-1}$ and obtain
$$
\mathbb{E}[U_j-U_{j-1}\mid v_1,\ldots,v_{j-1}]
\ge
\frac{|P_{j-1}|}{2\sqrt{B|T|n}}.
$$ 
Next, since each revealed edge can be the bottleneck of at most $B$ pairs from $P$, we have
$$
|P_{j-1}|
\ge
|P|-BU_{j-1}.
$$
Therefore
$$
\mathbb{E}[U_j-U_{j-1}\mid v_1,\ldots,v_{j-1}]
\ge
\frac{|P|-BU_{j-1}}{2\sqrt{B|T|n}}.
$$
Taking total expectation and writing $u_j:=\mathbb{E}[U_j]$, we obtain
$$
u_j-u_{j-1}
\ge
\frac{|P|-Bu_{j-1}}{2\sqrt{B|T|n}}.
$$

Let
$$
M:=\frac{|P|}{B},
\qquad
\lambda:=\frac{1}{2}\sqrt{\frac{B}{|T|n}}.
$$
Then the recurrence becomes
$$
u_j-u_{j-1}\ge \lambda(M-u_{j-1}),
$$
or equivalently,
$$
M-u_j\le (1-\lambda)(M-u_{j-1}).
$$
Iterating yields
$$
M-u_{\alpha k}\le (1-\lambda)^{\alpha k} M,
$$
and hence
$$
u_{\alpha k}\ge M\bigl(1-(1-\lambda)^{\alpha k}\bigr).
$$

Since $B\le |P|\le |T|n$, we have $\lambda\in(0,1/2]$. Therefore
$$
1-(1-\lambda)^{\alpha k}\ge \Omega(\min\{1,\alpha k\lambda\}).
$$
Substituting the definitions of $M$ and $\lambda$ gives
$$
u_{\alpha k}
\ge
\Omega\!\left(
\min\left\{
\frac{|P|}{B},
\alpha k\cdot \frac{|P|}{\sqrt{B|T|n}}
\right\}
\right)
=
\Omega\!\left(
\min\left\{
\frac{|P|}{B},
\alpha k\sqrt{\frac{|P|^2}{B|T|n}}
\right\}
\right).
$$
Since $u_{\alpha k}=\mathbb{E}\!\left[\left|\bigcup_{i=1}^{\alpha k} R(v_i)\right|\right]$, this proves the lemma.
\end{proof}

\subsubsection{Proof of \Cref{lem:ror}}

\lemthree*

\begin{proof}
Consider one iteration of \Cref{alg:ROR}. Let $T_{\mathrm{old}}$ and $Q_{\mathrm{old}}$ be the values stored in line~\ref{line:ROR-old}, and let $(S',B)$ be the output of $\textsc{High-Residual-BC}(Q_{\mathrm{old}},I)$. By \Cref{lem:high-bc-threshold},
$$
B\le \frac{|Q_{\mathrm{old}}|}{I}.
$$

After updating $S$ in line~\ref{line:ROR-update-S}, the algorithm recomputes $T$ and $Q$ in lines~\ref{line:ROR-update-T} and~\ref{line:ROR-update-Q}. Since residual degrees can only decrease when $S$ grows, the new set $T$ is contained in $T_{\mathrm{old}}$. Moreover, by the definition of $Q$ in line~\ref{line:ROR-update-Q}, every pair in the new $Q$ is a residual pair after deleting the pairs hit by $S'$. Hence
$$
Q\subseteq Q_{\mathrm{old}}\setminus P(S').
$$
Therefore, every edge lies on at most $B$ shortest paths induced by pairs in $Q$.

\Cref{alg:ROR} has two possible exits. If we run out of high-degree vertices, then $T=\emptyset$, which means
$$
\deg_{P\setminus P(S)}(v)<d
\qquad\text{for all }v\in V.
$$
Then the first conclusion is true. 

Otherwise, the algorithm returns through the if-branch in line~\ref{line:ROR-large}. Thus, in that iteration,
$|T|>|T_{\mathrm{old}}|/2$ and $|Q|>|Q_{\mathrm{old}}|/2$, and the query vertices are produced by the call to \Cref{alg:RWRBC} immediately following this line. Applying \Cref{lem:rwrbc} to the pair family $Q$ with betweenness centrality at most $B$ yields
$$
\mathbb{E}\!\left[\left|\bigcup_{i=1}^{\alpha k} R(v_i)\right|\right]
\ge
\Omega\!\left(
\min\left\{
\frac{|Q|}{B},
\alpha k\sqrt{\frac{|Q|^2}{B|T|n}}
\right\}
\right).
$$
In particular, the first term satisfies
$$
\frac{|Q|}{B} > \frac{|Q_{\mathrm{old}}|}{2} \cdot \frac{I}{|Q_{\mathrm{old}}|}
=
\frac{I}{2}.
$$
Moreover, every $u\in T$ has $\deg_{P\setminus P(S)}(u)\ge d$, and $Q$ contains exactly the pairs in $P\setminus P(S)$ whose first endpoint lies in $T$. Hence
$$
|Q|\ge d|T|.
$$
Therefore, the second term can be simplified as
$$
\alpha k\sqrt{\frac{|Q|^2}{B|T|n}}
\ge
\alpha k\sqrt{\frac{|Q|^2I}{|Q_{\mathrm{old}}||T|n}}
>
\alpha k\sqrt{\frac{|Q|I}{2|T|n}}
\ge
\frac{\alpha k}{\sqrt{2}}\sqrt{\frac{dI}{n}}.
$$
This proves Conclusion~2.

Next, we bound the number of iterations. Whenever an iteration does not return, the condition in line~\ref{line:ROR-large} is false. Therefore at least one of $|T|$ or $|Q|$ drops by a factor of $2$. Since initially $|T|\le n$ and $|Q|\le n^2$, this can happen only $O(\log n)$ times. In each iteration, we add at most $I$ new vertices. Thus, $|S|=O(I\log n)$.
\end{proof}

\subsubsection{Proof of \Cref{lem:ir}}

\lemfour*

\begin{proof}
For each iteration index $j\ge 0$, let $S_j$ be the cumulative set of vertices added to $S$ so far by the \textsc{ReduceOrReveal} step. We maintain an upper bound $D_j$ on
$$
\max_{v\in V}\deg_{\mathcal P\setminus P(S_j)}(v),
$$
starting from $S_0=\emptyset$ and $D_0=n$.

At the $j$-th iteration, following lines~\ref{line:IR-d} and~\ref{line:IR-ROR}, define
$$
d_j:=\frac{k^2D_j^2n}{\beta^2(\alpha k)^2I},
$$
and call
$$
\textsc{ReduceOrReveal}(\mathcal P\setminus P(S_j),d_j,I).
$$

After the update $S\gets S\cup S'_j$ in line~\ref{line:IR-update-S}, if the test in line~\ref{line:IR-check} is true, i.e.,
$$
\max_{v\in V}\deg_{\mathcal P\setminus P(S_j\cup S'_j)}(v)\ge d_j,
$$
then Conclusion~1 of \Cref{lem:ror} fails, and hence Conclusion~2 must hold. Therefore the returned tuple satisfies
$$
\mathbb{E}\!\left[\left|\bigcup_{i=1}^{\alpha k} R(v_i)\right|\right]
\ge
\Omega\!\left(
\min\left\{
I,
\alpha k\sqrt{\frac{d_jI}{n}}
\right\}
\right).
$$
By the choice of $d_j$,
$$
\alpha k\sqrt{\frac{d_jI}{n}}
=
\alpha k\sqrt{
\frac{k^2D_j^2n}{\beta^2(\alpha k)^2I}\cdot \frac{I}{n}
}
=
\frac{kD_j}{\beta}.
$$
Since $D_j\le n$ and $I\ge kn/\beta$ (from the assumption), we have
$$
I\ge \frac{kD_j}{\beta}.
$$
Thus, the lower bound simplifies to
$$
\Omega\!\left(\frac{kD_j}{\beta}\right).
$$
Residual degrees can only decrease as $S$ grows, so if $S$ denotes the final set returned by the algorithm, then
$$
\hat{d}=\max_{v\in V}\deg_{\mathcal P\setminus P(S)}(v)\le D_j.
$$
Hence, the returned tuple reveals
$$
\Omega\!\left(\frac{k\hat{d}}{\beta}\right)
$$
edges in expectation, as desired.

Otherwise, by Conclusion~1 of \Cref{lem:ror}, after augmenting $S_j$ by the returned set $S'_j$, every residual degree is below $d_j$. In this case, following lines~\ref{line:IR-update-S} and~\ref{line:IR-update-D}, set
$$
S_{j+1}=S_j\cup S'_j
\qquad\text{and}\qquad
D_{j+1}=d_j.
$$
Then we have
$$
\max_{v\in V}\deg_{\mathcal P\setminus P(S_{j+1})}(v)<D_{j+1}.
$$
So the invariant
$$
\max_{v\in V}\deg_{\mathcal P\setminus P(S_j)}(v)\le D_j
$$
is preserved.

Finally, by the first hypothesis and the fact that $D_j\le n$, we have
$$
D_{j+1}
=
d_j
=
D_j\cdot \frac{k^2D_jn}{\beta^2(\alpha k)^2I}
\le
D_j\cdot \frac{k^2n^2}{\beta^2(\alpha k)^2I}
\le
\frac{D_j}{2}.
$$
Hence the sequence $D_j$ decreases geometrically. After $O(\log n)$ iterations we must have $D_j<1$, and since all residual degrees are integers this implies
$$
\max_{v\in V}\deg_{\mathcal P\setminus P(S_j)}(v)=0.
$$
In that case, the desired lower bound is trivial because $S=S_j$ and $\hat{d}=0$.

Moreover, by \Cref{lem:ror}, each call to \textsc{ReduceOrReveal} adds $O(I\log n)$ vertices to $S$. Since the number of calls is $O(\log n)$, the final set satisfies
$$
|S|=O(I\log^2 n).
$$
Thus, the claim is true in all cases.
\end{proof}

\subsubsection{Final Analysis}

Recall the bound we aim to prove: 

$$
\frac{\OPT_k}{\mathbb{E}(\ALG_{\alpha k})}=
\tilde{O}\!\left(
\min\left\{
\frac{n}{\alpha k},
\max\left(
\sqrt{\frac{n}{\alpha}},
\frac{n^{2/3}}{\alpha k^{1/3}}
\right)
\right\}
\right).
$$

\noindent We first state a simple sampling lemma.

\begin{lemma}
\label{lem:witness}
Let $U\subseteq V$ and let $F\subseteq E$. Suppose that for every $e\in F$, there exists a vertex $u_e\in U$ such that $e\in R(u_e)$. If $\{v_1,\dots,v_{\alpha k}\}$ are sampled uniformly at random from $U$, then
$$
\mathbb{E}\!\left[
\left|
\left(\bigcup_{i=1}^{\alpha k} R(v_i)\right)\cap F
\right|
\right]
\ge
\Omega\!\left(
\min\left\{
1,\frac{\alpha k}{|U|}
\right\}
|F|
\right).
$$
\end{lemma}

\begin{proof}
For each edge $e\in F$, fix a witness $u_e\in U$ such that $e\in R(u_e)$. Let $X_e$ be the indicator of the event that $e$ is revealed by at least one of the $\alpha k$ samples. If one of the samples equals $u_e$, then $e$ is certainly revealed. Therefore
$$
\Pr[X_e=1]
\ge
1-\left(1-\frac{1}{|U|}\right)^{\alpha k}
=
\Omega\!\left(
\min\left\{
1,\frac{\alpha k}{|U|}
\right\}
\right).
$$
By linearity of expectation,
$$
\mathbb{E}\!\left[
\left|
\left(\bigcup_{i=1}^{\alpha k} R(v_i)\right)\cap F
\right|
\right]
=
\sum_{e\in F}\Pr[X_e=1]
\ge
\Omega\!\left(
\min\left\{
1,\frac{\alpha k}{|U|}
\right\}
|F|
\right).
$$
\end{proof}


\noindent Now we are ready to proceed with the proof of the final bound. Let $\{x_1,\dots,x_k\}$ be an optimal choice of $k$ vantage points, and define
$$
F^\star:=\bigcup_{i=1}^k R(x_i).
$$
Then $|F^\star|=\OPT_k$. Take an edge $e\in F^\star$. Since $e\in F^\star$, there exist an index $i(e)\in[k]$ and a vertex $t(e)$ with $(x_{i(e)},t(e))\in\mathcal P$ such that $e$ is the bottleneck edge on the shortest path from $x_{i(e)}$ to $t(e)$, that is,
$$
e=b(x_{i(e)},t(e)).
$$
Now define $F_{\mathrm{hit}}$ as the set of edges in $F^\star$ that would also be revealed by vertices in $S$, i.e., 
$$
F_{\mathrm{hit}}
:=
\{e\in F^\star\mid \SPV(x_{i(e)},t(e))\cap S\neq \emptyset\},
$$
and let $F_{\mathrm{res}}$ be the remaining edges:
$$
F_{\mathrm{res}}
:=
F^\star\setminus F_{\mathrm{hit}}.
$$

We first verify the assumptions of \Cref{lem:ir} for the parameter choice of $I$ and $\beta$ in \Cref{alg:final}. Recall that
$$
I=\lceil \alpha k\beta\rceil=\Theta(\alpha k\beta),
$$
and
$$
\beta
=
\tilde{\Theta}\!\left(
\max\left\{
\sqrt{\frac{n}{\alpha}},
\frac{n^{2/3}}{\alpha k^{1/3}}
\right\}
\right).
$$
By choosing the hidden constant in the definition of $\beta$ sufficiently large, we have
$$
I\ge \alpha k\beta \ge \frac{kn}{\beta}
$$
because $\beta\ge \tilde{\Omega}(\sqrt{n/\alpha})$, and
$$
\frac{k^2n^2}{\beta^2(\alpha k)^2I}
\le
\frac{n^2}{\alpha^3\beta^3k}
\le
\frac12
$$
because $\beta\ge \tilde{\Omega}(n^{2/3}/(\alpha k^{1/3}))$. Therefore \Cref{lem:ir} yields
$$
\mathbb{E}\!\left[\left|\bigcup_{i=1}^{\alpha k} R(v'_i)\right|\right]
\ge
\tilde{\Omega}\!\left(
\frac{k\max_{v\in V}\deg_{\mathcal{P}\setminus P(S)}(v)}{\beta}
\right).
$$

We now analyze the three branches of \Cref{alg:final}: Branch 1 is the case $X=1$ in line~\ref{line:final-branch1}, Branch 2 is the case $X=2$ in line~\ref{line:final-branch2}, and Branch 3 is the case $X=3$ in line~\ref{line:final-branch3}.

\paragraph{Branch 1: sampling from $S$.}
Let $e\in F_{\mathrm{hit}}$, and let
$$
(s,t):=(x_{i(e)},t(e)).
$$
By definition, $\SPV(s,t)\cap S\neq\emptyset$. Pick any vertex $u\in \SPV(s,t)\cap S.$ Since $u\in \SPV(s,t)$, the unique shortest path from $s$ to $t$ is the concatenation of the two shortest subpaths from $s$ to $u$ and from $u$ to $t$, so the edge $e$ lies on at least one of the two edge sets $\SPE(s,u)$ and $\SPE(u,t)$. Suppose first that $e\in \SPE(s,u)$. Since $e$ is the unique minimum-capacity edge among the edges in $\SPE(s,t)$, it is also the unique minimum-capacity edge among the edges in $\SPE(s,u)$. Hence
$$
e=b(s,u),
$$
and therefore $e\in R(u)$. The case $e\in \SPE(u,t)$ is similar and gives $e=b(u,t)$. Thus every edge $e\in F_{\mathrm{hit}}$ can be revealed by a vertex in $S$, called its witness.
Since $|S|=O(I\log^2 n)=O(\alpha k\beta\log^2 n)$, \Cref{lem:witness} with $U=S$ gives the following bound. Recall that
$
\ALG_{\alpha k}:=\left|\bigcup_{i=1}^{\alpha k} R(v_i)\right|.
$
Then
$$
\mathbb{E}\!\left[\ALG_{\alpha k}\mid\text{Branch 1}\right]
\ge
\tilde{\Omega}\!\left(\frac{|F_{\mathrm{hit}}|}{\beta}\right).
$$

\paragraph{Branch 2: using the vantage points returned by \Cref{alg:IR}.}
If $e\in F_{\mathrm{res}}$, then by definition
$$
(x_{i(e)},t(e))\in \mathcal{P}\setminus P(S).
$$
For each fixed $i\in [k]$, the number of residual edges revealed by $x_i$ is at most $\deg_{\mathcal{P}\setminus P(S)}(x_i)$. Consequently,
$$
|F_{\mathrm{res}}|
\le
\sum_{i=1}^k \deg_{\mathcal{P}\setminus P(S)}(x_i)
\le
k\max_{v\in V}\deg_{\mathcal{P}\setminus P(S)}(v).
$$
Using the guarantee from \Cref{lem:ir}, we obtain
$$
\mathbb{E}\!\left[\ALG_{\alpha k}\mid \text{Branch 2}\right]
\ge
\tilde{\Omega}\!\left(
\frac{k\max_{v\in V}\deg_{\mathcal{P}\setminus P(S)}(v)}{\beta}
\right)
\ge
\tilde{\Omega}\!\left(
\frac{|F_{\mathrm{res}}|}{\beta}
\right).
$$

\paragraph{Branch 3: sampling from $V$.}
Every edge $e$ in $F^\star$ can be revealed by a query at one of its witnesses
$x_i$. Hence \Cref{lem:witness} with $U=V$ yields
$$
\mathbb{E}\!\left[\ALG_{\alpha k}\mid \text{Branch 3}\right]
\ge
\Omega\!\left(
\min\left\{
1,\frac{\alpha k}{n}
\right\}
|F^\star|
\right)
=
\Omega\!\left(
\frac{\alpha k}{n}\OPT_k
\right),
$$
where the last equality uses $\alpha k\le n$.

Since each branch is chosen with probability $1/3$, combining the three estimates gives
$$
\mathbb{E}\!\left[\ALG_{\alpha k}\right]
\ge
\tilde{\Omega}\!\left(
\frac{|F_{\mathrm{hit}}|+|F_{\mathrm{res}}|}{\beta}
\right)
+
\Omega\!\left(
\frac{\alpha k}{n}\OPT_k
\right).
$$
Because $F_{\mathrm{hit}}\cup F_{\mathrm{res}}=F^\star$ and $|F^\star|=\OPT_k$, this simplifies to
$$
\mathbb{E}\!\left[\ALG_{\alpha k}\right]
\ge
\tilde{\Omega}\!\left(
\max\left\{
\frac{\OPT_k}{\beta},
\frac{\alpha k}{n}\OPT_k
\right\}
\right).
$$
Equivalently, the competitive ratio is
$$
\frac{\OPT_k}{\mathbb{E}[\ALG_{\alpha k}]}
\le
\tilde{O}\!\left(
\min\left\{
\beta,
\frac{n}{\alpha k}
\right\}
\right)
=
\tilde{O}\!\left(
\min\left\{
\frac{n}{\alpha k},
\max\left(
\sqrt{\frac{n}{\alpha}},
\frac{n^{2/3}}{\alpha k^{1/3}}
\right)
\right\}
\right),
$$
as claimed.

\subsection{Running Time Analysis}

We describe an implementation of \Cref{alg:final} with running time $\tilde{O}(n(n+m))$.

First, for every vertex $s \in V$, compute the shortest-path tree $T_s$ rooted at $s$ and spanning the connected component containing $s$. Since shortest paths are unique, the path $\SPV(s,t)$ is exactly the root-to-$t$ path in $T_s$. For every rooted tree $T_s$, compute one Euler tour and for each vertex $v$ store $\tin_s(v),\tout_s(v)$, which refer to the interval endpoints for the subtree. Thus $v\in \SPV(s,t)$ if and only if $t$ lies in the subtree of $v$ in $T_s$, equivalently $\tin_s(t)\in [\tin_s(v),\tout_s(v)]$. The $n$ shortest-path tree computations take $\tilde{O}(n(n+m))$ time by running a standard single-source shortest-path algorithm from every source.

The entire algorithm makes several calls to \Cref{alg:high-bc}. Inside each call of \Cref{alg:high-bc}, a current pair family $P_{cur}$ is maintained, and we will simply write $P$ instead of $P_{cur}$ in the following for the sake of simplicity. Define the \emph{target} set of a source $s$ with respect to $P$ by
$$
N_P(s):=\{t\in V\mid (s,t)\in P\}.
$$
We maintain $N_P(s)$ in the order induced by the Euler tour of $T_s$. More concretely, for each $s$, we store the values ${\tin_s(t):t\in N_P(s)}$ in a balanced binary search tree. Here this means a balanced binary search tree with height $O(\log n)$ that supports listing all targets in an interval and deleting all targets in an interval, each in $O((z+1)\log n)$ time, where $z$ is the number of targets in the interval. It also gives $\deg_P(s)=|N_P(s)|$ in $O(1)$ time. Such a data structure is given in Section~2.1.4 of \cite{henzinger1999randomized}. 

We next explain how to implement \Cref{alg:high-bc}. At the beginning of one call with pair family $P$, compute all values $\bc_P(v)$ and $\bc_P(e)$ as follows. For each source $s$, a postorder traversal of $T_s$ gives, for every vertex $v$, the number of targets both in $N_P(s)$ and in the subtree of $v$ in $T_s$. This number is added to $\bc_P(v)$ and, if $v\neq s$, to $\bc_P(e)$, where $e=(p,v)$ and $p$ is the parent of $v$ in $T_s$. Over all roots $s$, this takes $O(n^2)$ time. We store the current values $\bc_P(v)$ in a max-heap supporting extraction of a maximum key and key decreases in $O(\log n)$ time.

Now suppose that \Cref{alg:high-bc} selects a vertex $x=v_i$ in line~\ref{line:HBC6}. The pairs deleted from the current pair family $P$ are

$$
D:=\{(s,t)\in P\mid x\in \SPV(s,t)\}.
$$
Using the Euler intervals, for every source $s$ we find exactly these deleted targets by listing and deleting all targets of $N_P(s)$ whose Euler-tour values lie in the interval $[\tin_s(x),\tout_s(x)]$ of $T_s$. During this deletion, for each vertex $s$, we compute
$$
A_s:=\big|\{t\mid (s,t)\in D\}\big|,
\qquad
B_s:=\big|\{t\mid (t,s)\in D\}\big|.
$$
The total time for this process is $O((|D|+n)\log n)$. Here the $|D|\log n$ term accounts for the deleted targets, and the $n\log n$ term comes from making one interval query for each of the $n$ sources, even when the corresponding interval contains no target. 

We will use $A$ and $B$ to update the residual betweenness centrality values. Put weight
$$
w(y):=A_y+B_y
$$
on each vertex $y$ of the shortest-path tree $T_x$, and let $W(y)$ be the total weight in the subtree of $y$ in $T_x$. These values are computed by one postorder traversal of $T_x$. For every vertex $y\neq x$, decrease $\bc_P(y)$ by $W(y)$, and decrease $\bc_P(x)$ by $|D|$. For every parent edge $e=(p,y)$ in $T_x$, decrease $\bc_P(e)$ by $W(y)$, where $p$ is the parent of $y$ in $T_x$. These updates are correct because, for every deleted pair $(s,t)\in D$, the unique shortest path from $s$ to $t$ contains $x$, and hence is the union of the two shortest paths from $x$ to $s$ and from $x$ to $t$. Therefore a vertex $y\neq x$, or an edge with child endpoint $y$ in $T_x$, lies on the shortest $xs$-path, respectively $xt$-path, exactly when $s$, respectively $t$, lies in the subtree of $y$ in $T_x$.

In one call to \Cref{alg:high-bc}, each ordered pair of $P$ is deleted at most once, and at most $n$ vertices are selected. Hence the total time for all interval deletions, all postorder traversals, and all heap updates in this call is $\tilde{O}(n^2)$. 

By \Cref{lem:ror,lem:ir}, \Cref{alg:final} makes only $O(\log^2 n)$ calls to \Cref{alg:high-bc}. We now account for the remaining operations. Each line of \Cref{alg:final} takes at most linear time, except for calling \Cref{alg:IR}. In \Cref{alg:IR}, the while loop has $O(\log n)$ iterations by \Cref{lem:ir}. In each iteration, the call to \Cref{alg:ROR} is the only non-linear operation. The maximum residual degree is computed by scanning all vertices and using the maintained values $\deg_P(x)=|N_P(x)|$, which takes $O(n)$ time. In \Cref{alg:ROR}, the while loop has $O(\log n)$ iterations by \Cref{lem:ror}. Each iteration forms $T$ by scanning all vertices and forms $Q$ by scanning the currently maintained pair family, taking $O(n^2)$ time, and then makes one call to \Cref{alg:high-bc}. Finally, \Cref{alg:RWRBC} only samples the required query vertices, and hence takes linear time. Since one call to \Cref{alg:high-bc} takes $\tilde{O}(n^2)$ time as proved above, and there are only $O(\log^2 n)$ such calls, the total time after the shortest-path preprocessing is $\tilde{O}(n^2)$. Thus the total running time is
$$
\tilde{O}(n(n+m)),
$$
whose bottleneck is the shortest-path tree computation.

\section{Lower Bound}\label{sec:lb}

In this section we prove the lower bound part of \Cref{thm:main}, which we restate here:

\begin{theorem}
    For any positive integers $\alpha,k$ with $\alpha k \le n$, and for every randomized non-adaptive algorithm, there exists an $n$-vertex weighted graph with unique shortest paths between reachable pairs, together with a capacity assignment, such that when the algorithm is run on this instance, it has an expected competitive ratio of
$$
\frac{\OPT_k}{\mathbb{E}(\ALG_{\alpha k})}=
\tilde{\Omega}\!\left(
\min\left\{
\frac{n}{\alpha k},
\max\left(
\sqrt{\frac{n}{\alpha}},
\frac{n^{2/3}}{\alpha k^{1/3}}
\right)
\right\}
\right).
$$
\end{theorem}




\subsection{Construction}

The construction is illustrated in \Cref{fig:full}.

\begin{figure}[h!]
    \centering
    \includegraphics[width=1\textwidth]{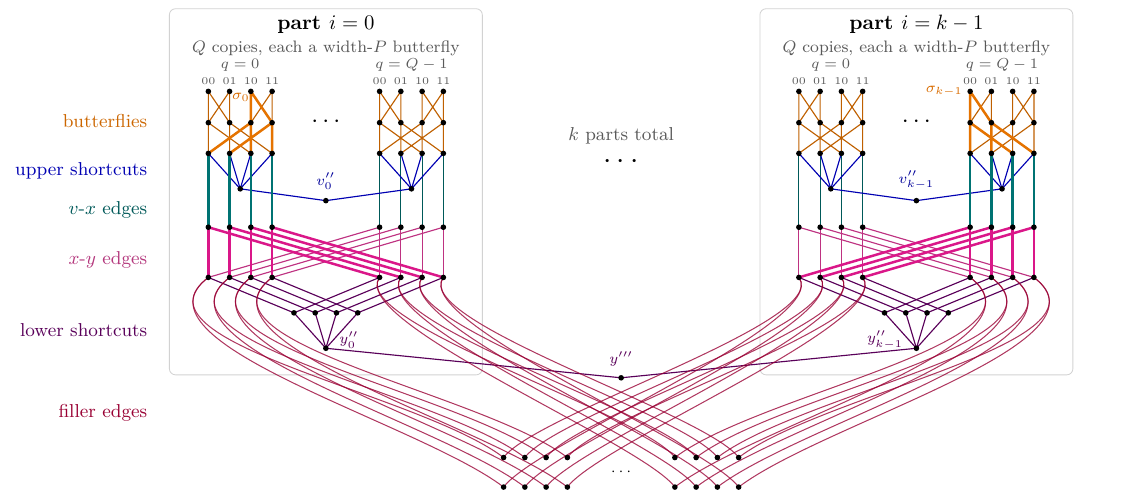}
    \caption{The lower-bound construction for the general multi-query setting.}
    \label{fig:full}
\end{figure}

The graph has $k$ parts. In each part, the index $q\in\Zx{Q}$ labels one of $Q$ \emph{copies}. For each copy index $q$, we add a butterfly graph and the associated $x$- and $y$-vertices defined below. The \emph{filler} vertices are connected to the $y$-layers of all parts and therefore do not belong to any single part. The graph also has two \emph{shortcut structures}: an upper shortcut structure, which is part-specific, and a lower shortcut structure, which is shared across parts. 

All edge weights below are first specified as \emph{base weights}. After all base weights have been assigned, we add distinct positive infinitesimal perturbations to the weights, chosen small enough that every strict base-weight comparison used in the proof remains unchanged. These perturbations make all shortest paths unique. The perturbations are fixed deterministically and are not part of the probability space used in the analysis.

\paragraph{Butterflies.}
For each part $i\in\Zx{k}$ and each copy index $q\in\Zx{Q}$, we add a width-$P$ butterfly graph with $h+1$ levels, where $P=2^h$. The parameters $Q,P,h$ are set later. The butterfly vertices are
$$
v_{i,\ell,p,q},
\qquad
 i\in\Zx{k},\quad \ell\in\{0,1,\ldots,h\},\quad p\in\Zx{P},\quad q\in\Zx{Q}.
$$
We refer to $\ell$ as the \emph{level}, to $p$ as the \emph{address}, and to $q$ as the \emph{copy index}. The vertices with $\ell=0$ are the \emph{top-level} butterfly vertices, and the vertices with $\ell=h$ are the \emph{bottom-level} butterfly vertices.

Recall that a butterfly graph is defined as follows. For every $v_{i,\ell,p,q}$ with $\ell<h$, add the two butterfly edges
$$
(v_{i,\ell,p,q},v_{i,\ell+1,p,q})
\qquad\text{and}\qquad
(v_{i,\ell,p,q},v_{i,\ell+1,p\oplus 2^\ell,q}),
$$
each with base weight $1$. Here $\oplus$ denotes bitwise XOR on the binary representations of the zero-based addresses.

For every top-level vertex $v_{i,0,p,q}$ and every bottom address $p'\in\Zx{P}$, there is a unique \emph{canonical downward path} in the butterfly from $v_{i,0,p,q}$ to $v_{i,h,p',q}$: at level $\ell$, the path takes the XOR edge exactly when the $\ell$-th bit of $p\oplus p'$ is $1$. The reverse of such a path is called a \emph{canonical upward path}.

For later use, let $D^V_{i,p,q}$ be the set of butterfly vertices that lie on a canonical downward path from $v_{i,0,p,q}$ to $v_{i,h,p',q}$ for some $p'\in\Zx{P}$. We call this the downward sub-butterfly rooted at $v_{i,0,p,q}$. Let $D^E_{i,p,q}$ be the set of butterfly edges on these paths. For every butterfly vertex $v_{i,\ell,p,q}$, let $I_{i,\ell,p,q}$ be the set of top-level vertices whose canonical downward paths can reach $v_{i,\ell,p,q}$, and let $O_{i,\ell,p,q}$ be the set of bottom-level vertices reachable from $v_{i,\ell,p,q}$ by canonical downward suffixes. Thus
$$
|I_{i,\ell,p,q}|=2^\ell
\qquad\text{and}\qquad
|O_{i,\ell,p,q}|=P/2^\ell.
$$

\paragraph{$x$-layer, $y$-layer, and filler vertices.}
For each part, there are two layers in addition to the butterfly vertices: the $x$-layer and the $y$-layer. Moreover, the cross-part filler vertices are connected to the $y$-layer vertices in every part.

As with the butterfly vertices, all these vertices have an address and a copy index. We say that two vertices have the same address if their $p$-coordinates are equal, and have the same copy index if their $q$-coordinates are equal.

For the $x$-layer, add vertices
$$
x_{i,p,q},
\qquad
 i\in\Zx{k},\quad p\in\Zx{P},\quad q\in\Zx{Q}.
$$
Here $p$ is the address and $q$ is the copy index. For every $x_{i,p,q}$, add the $v$-$x$ edge
$$
(v_{i,h,p,q}, x_{i,p,q})
$$
with base weight $1$. These edges form a matching between the $x$-layer and the bottom-level butterfly vertices.

For the $y$-layer, add vertices
$$
y_{i,p,q},
\qquad
 i\in\Zx{k},\quad p\in\Zx{P},\quad q\in\Zx{Q}.
$$
For every $y_{i,p,q}$, add the $x$-$y$ edge
$$
(x_{i,p,q'}, y_{i,p,q})
\qquad
q'\in\Zx{Q}
$$
with base weight $1$.

Finally, for the filler vertices, introduce a parameter $\tau$, which is set later, and add vertices
$$
f_{t,p,q},
\qquad
t\in\Zx{\tau},\quad p\in\Zx{P},\quad q\in\Zx{Q}.
$$
Here $t$ is the filler index. For every filler vertex $f_{t,p,q}$, add the filler edge
$$
(y_{i,p,q},f_{t,p,q})
\qquad
i\in\Zx{k}
$$
with base weight $1$.

\paragraph{Shortcut structures.}
There are two shortcut structures in the graph: an upper shortcut structure and a lower shortcut structure.

The upper shortcut structure is built on the bottom level of the butterflies and is independent for each part. The upper shortcut vertices are
$$
v'_{i,q}
\qquad
 i\in\Zx{k},\quad q\in\Zx{Q}
$$
and
$$
v''_i
\qquad
 i\in\Zx{k}.
$$
For every $v'_{i,q}$, add upper shortcut edges
$$
(v_{i,h,p,q}, v'_{i,q})
\qquad
p\in\Zx{P}
$$
with base weight $\varepsilon_2$. For every $v''_i$, add upper shortcut edges
$$
(v'_{i,q}, v''_i)
\qquad
 q\in\Zx{Q}
$$
with base weight $\varepsilon_1$. Here $\varepsilon_1$ and $\varepsilon_2$ are fixed constants satisfying
$$
0<\varepsilon_1\ll \varepsilon_2\ll 1.
$$

The lower shortcut structure is built on the $y$-layer and is shared across parts through the vertex $y'''$. The lower shortcut vertices are
$$
y'_{i,p}
\qquad
 i\in\Zx{k},\quad p\in\Zx{P},
$$
$$
y''_i
\qquad
i\in\Zx{k},
$$
and
$$
y'''.
$$

\noindent For every $y'_{i,p}$, add lower shortcut edges
$$
(y_{i,p,q'}, y'_{i,p})
\qquad
q'\in\Zx{Q}.
$$
For every $y''_i$, add lower shortcut edges
$$
(y'_{i,p},y''_i)
\qquad
p\in\Zx{P}
$$
and add the lower shortcut edge
$$
(y''_i,y''').
$$
All these lower shortcut edges have base weight $\varepsilon_1$.

\paragraph{Capacity distribution.}

For each part $i\in\Zx{k}$, choose independently and uniformly at random
$$
\sigma_i=(p_i^\star,q_i^\star)\in\Zx{P}\times\Zx{Q}.
$$
The hidden optimal query vertex in part $i$ will be the top-level butterfly vertex $v_{i,0,p_i^\star,q_i^\star}$. Given $\sigma=(\sigma_i)_{i\in\Zx{k}}$, we first designate the following edges as \emph{high}:
\begin{enumerate}[label=\arabic*.]
    \item the butterfly edges in $D^E_{i,p_i^\star,q_i^\star}$ for all $i\in\Zx{k}$;
    \item the $v$-$x$ edges $(v_{i,h,p',q_i^\star},x_{i,p',q_i^\star})$ for all $i\in\Zx{k}$ and $p'\in\Zx{P}$;
    \item the $x$-$y$ edges $(x_{i,p',q_i^\star},y_{i,p',q'})$ for all $i\in \Zx{k}$, $p'\in\Zx{P}$, and $q'\in\Zx{Q}$.
\end{enumerate}
Base capacities are then assigned according to the following table.

\begin{center}
\begin{tabular}{c|c}
\toprule
base capacity & edges \\
\midrule
$1$ & $(y''_i,y''')$ \\
$2$ & $(v'_{i,q},v''_i)$, $(y'_{i,p},y''_i)$ \\
$3$ & $(v_{i,h,p,q},v'_{i,q})$, $(y_{i,p,q'},y'_{i,p})$ \\
$4$ & all non-high $v$-$x$ edges \\
$5+\frac{h-\ell}{h}$ & all non-high butterfly edges between butterfly levels $\ell$ and $\ell+1$, for $\ell=0,\ldots,h-1$ \\ 
$6$ & all non-high $x$-$y$ edges \\
$7$ & all filler edges $(y_{i,p,q},f_{t,p,q})$ \\
$8$ & high $v$-$x$ edges \\
$9$ & high butterfly edges and high $x$-$y$ edges \\
\bottomrule
\end{tabular}
\end{center}

Within each base-capacity class, we add distinct infinitesimal perturbations to make all capacities distinct, small enough that all comparisons between different base capacities remain unchanged. These perturbations are fixed deterministically after $\sigma$ is chosen.


\paragraph{Size and padding.}
We first define the parameters. Let
$$
\tau:=\left\lceil\max\left\{k,\min\left\{\sqrt{\frac n\alpha},(kn)^{1/3}\right\}\right\}\right\rceil
$$
and
$$
Q:=\left\lceil\frac{\tau}{k}\right\rceil.
$$
Since $\tau\ge k$, we have $kQ=\Theta(\tau)$.

Choose $\Lambda=\Theta(\log^2 n)$ with a sufficiently large hidden constant and set
$$
N_0:=\left\lfloor\frac{n}{\Lambda}\right\rfloor.
$$

Let $P=2^h$ be a power of two satisfying
$$
\frac{N_0}{4Q\tau}\le P\le \frac{N_0}{Q\tau},
$$
provided that $N_0/(Q\tau)$ is at least a sufficiently large constant. If $N_0/(Q\tau)$ is smaller than this constant, then the target lower bound is only polylogarithmic: a direct check\footnote{In the regime $\tau=k$ this follows from $Q\tau=\Theta(k)$; in the regime $\tau=\sqrt{n/\alpha}$, the condition $N_0/(Q\tau)=O(1)$ implies $\alpha k=\log^{O(1)}n$ and the regime condition gives $n/(\alpha\tau)=\tau\le \alpha k$; and in the regime $\tau=(kn)^{1/3}$ the condition can occur only when $n=\log^{O(1)}n$. In that case, a star on $n$ vertices with distinct edge capacities already gives a constant worst-case ratio, which is enough for the claimed $\tilde{\Omega}$ bound.} from the definition of $\tau$ gives
$$
\frac{n}{\alpha\tau}\le \log^{O(1)} n.
$$
 Hence, for the rest of the proof we assume $N_0/(Q\tau)$ is sufficiently large and use the construction above.

We now check the number of vertices before padding. The filler vertices contribute $PQ\tau$ vertices. The butterflies contain $O(kQP\log P)$ vertices, and the remaining non-filler vertices contribute $O(kQP+kP+kQ)$ vertices. Since $kQ=O(\tau)$, the non-filler vertices together have size $O(P\tau\log P)$. Therefore the total number of vertices is
$$
O(PQ\tau+P\tau\log P)\le O(N_0\log n).
$$
Taking the constant in $\Lambda$ large enough makes this graph have at most $n$ vertices.

To reach exactly $n$ vertices, we add the required number of isolated vertices. Since no reachable ordered pair contains isolated vertices, they do not affect the competitive-ratio lower bound. We therefore analyze the original graph below.

\subsection{Analysis}


\paragraph{Upper bound on $\mathbb{E}[\ALG_{\alpha k}]$.}

Fix a deterministic query set $S$ with $|S|\le \alpha k$ before the random vector $\sigma$ is chosen. We will show that 
$$
\mathbb{E}_{\sigma}[|R(S)|]\le\tilde{O}(\alpha k\tau),
$$
where the expectation is only over $\sigma$.

Let $E_{\mathrm{small}}$ be the set of edges with base capacity $1$ or $2$. Let $E_{\mathrm{large}}$ be the set of $x$-$y$ edges and filler edges. Let $E_{\mathrm{mid}}$ be the set of remaining edges. (These names refer to the sizes of the sets, rather than to the capacities of the edges they contain.)
Then
$$
|E_{\mathrm{small}}|\le O(kQ+kP)=O(\alpha k\tau)
$$
and
$$
|E_{\mathrm{mid}}|\le O(kQP\log P)=\tilde{O}(kQP).
$$
Here we use $kQ=O(\tau)$ and $P\le O(\alpha\tau)$. To see the latter, note that $P\le n/(Q\tau)\le nk/\tau^2$, since $Q\ge \tau/k$. By the definition of $\tau$, at least one of $\sqrt{n/\alpha}\le \tau$ and $(nk)^{1/3}\le \tau$ holds. In the first case, $nk/\tau^2\le \alpha k\le \alpha\tau$; in the second case, $nk/\tau^2\le \tau\le \alpha\tau$. Hence $P\le O(\alpha\tau)$.

We will prove the following single-query bound.
\begin{lemma}\label{lem:lower-bound-single-query}
For every fixed query vertex $u$,
$$
\mathbb{E}_{\sigma}[|R(u)\cap E_{\mathrm{large}}|]=O(\tau)
\qquad\text{and}\qquad
\mathbb{E}_{\sigma}[|R(u)\cap E_{\mathrm{mid}}|]=\tilde{O}(\tau+P).
$$
\end{lemma}

Before proving the lemma, we show that it implies the claimed bound of $
\mathbb{E}_{\sigma}[|R(S)|]\le\tilde{O}(\alpha k\tau)
$. Since $E_{\mathrm{small}}$ is counted globally, and since $E_{\mathrm{mid}}$ can be bounded either globally or by summing the single-query bound, we have
\begin{align*}
\mathbb{E}_{\sigma}[|R(S)|]
&\le |E_{\mathrm{small}}|
   +\min\left\{|E_{\mathrm{mid}}|,\sum_{u\in S}\mathbb{E}_{\sigma}[|R(u)\cap E_{\mathrm{mid}}|]\right\}
   +\sum_{u\in S}\mathbb{E}_{\sigma}[|R(u)\cap E_{\mathrm{large}}|]\\
&\le \tilde{O}\bigl(\alpha k\tau+
      \min\{kQP,\alpha k(\tau+P)\}+\alpha k\tau\bigr).
\end{align*}
It remains to bound the middle term. Since $P\le n/(Q\tau)$, we have $kQP\le kn/\tau$. If $\sqrt{n/\alpha}\le \tau$, then $n\le \alpha\tau^2$, and hence
$$
kQP\le \frac{kn}{\tau}\le \alpha k\tau.
$$
If instead $(nk)^{1/3}\le \tau$, then
$$
P\le \frac{n}{Q\tau}\le \frac{nk}{\tau^2}\le \tau,
$$
and hence $\alpha k(\tau+P)=O(\alpha k\tau)$. Therefore
$$
\mathbb{E}_{\sigma}[|R(S)|]\le\tilde{O}(\alpha k\tau).
$$

\begin{proof}[Proof of \Cref{lem:lower-bound-single-query}]

We do a case analysis over the type of the query vertex and the type of the target vertex. The vertex types are butterfly vertices, $x$-layer vertices, $y$-layer vertices, filler vertices, upper shortcut vertices, lower shortcut vertices, and the shared vertex $y'''$. Since there are only constantly many ordered pairs of types, it suffices to fix the query vertex $u$, fix the target type, and bound the number of possible bottleneck edges obtained as the target ranges over all vertices of that type.

\noindent\textbf{Case 1: two endpoints inside one part.}
Fix a part $i$ and suppose first that both endpoints are non-filler vertices associated with this part, not including $y'''$.

\begin{itemize}
\item Suppose the two endpoints are both butterfly vertices, or one endpoint is a butterfly vertex and the other is an $x$-layer vertex. 
\begin{itemize}
    \item If the two endpoints lie in the same copy $q$, the entire shortest path uses only butterfly edges, $v$-$x$ edges, and upper shortcut edges incident to $v'_{i,q}$. It never enters the $y$-layer, because doing so would add at least one extra base-weight-$1$ edge, while the shortcut through $v'_{i,q}$ connects bottom-level addresses at cost $2\varepsilon_2\ll 1$. For a fixed query vertex, the possible bottleneck edges of this form are contained in one width-$P$ butterfly copy, the $P$ matching $v$-$x$ edges, and the $P$ incident bottom shortcut edges. This gives $O(P\log P)$ possible edges in $E_{\mathrm{mid}}$ and no contribution to $E_{\mathrm{large}}$.
    \item Otherwise, the shortest path goes through $v''_i$ and hence has bottleneck in $E_{\mathrm{small}}$.
\end{itemize}

\item Suppose the two endpoints both lie in the union of the $x$-layer and the $y$-layer.
\begin{itemize}
    \item If the two endpoints have the same address, then for a fixed query vertex there are only $O(Q)$ relevant targets with that address. Thus these paths contribute $O(Q)=O(\tau)$ possible bottlenecks.
    \item Otherwise, using the lower shortcut through $y''_i$ adds only $O(\varepsilon_1)$ beyond the necessary incident edges, while any path avoiding $y''_i$ must use an upper shortcut or additional base-weight-$1$ edges. Hence the bottleneck lies in $E_{\mathrm{small}}$.
\end{itemize}

\item Suppose one endpoint is a butterfly vertex $v_{i,\ell,p,q}$ and the other endpoint is a $y$-layer vertex $y_{i,p',q'}$.
\begin{itemize}
    \item If $v_{i,h,p',q}\notin O_{i,\ell,p,q}$, the shortest path changes addresses through $y''_i$, so the bottleneck lies in $E_{\mathrm{small}}$. 
    \item Otherwise, the shortest path is
    $$
    v_{i,\ell,p,q}\rightsquigarrow v_{i,h,p',q}-x_{i,p',q}-y_{i,p',q'}.
    $$
    \begin{itemize}
    \item If $q\ne q_i^\star$, then the $v$-$x$ edge is non-high and has base capacity $4$, while all butterfly and $x$-$y$ edges on this path have base capacity at least $5$. Hence the bottleneck is this $v$-$x$ edge.
    \item If $q=q_i^\star$ and $v_{i,\ell,p,q}\in D^V_{i,p_i^\star,q_i^\star}$, then the butterfly and $x$-$y$ edges have base capacity $9$, while the $v$-$x$ edge has base capacity $8$, so again the bottleneck is the $v$-$x$ edge.
    \item Otherwise, the bottleneck edge is a deepest non-high butterfly edge on the butterfly suffix of the path.
    \end{itemize}
    Hence in all cases, the bottleneck edge belongs to $E_{\mathrm{mid}}$. If the query vertex is the butterfly endpoint, all such bottlenecks are contained in one butterfly copy together with the relevant $v$-$x$ edges, so there are $O(P\log P)$ possibilities. If the query vertex is the $y$-layer endpoint, then the address $p'$ is fixed: the non-hidden copies contribute only the $Q$ possible $v$-$x$ edges with that address, and the hidden copy contributes at most $O(P\log P)$ butterfly edges. Thus this subcase contributes $\tilde{O}(\tau+P)$ edges in $E_{\mathrm{mid}}$.
\end{itemize}

\item Suppose one endpoint is an upper shortcut vertex, namely $v'_{i,q}$ or $v''_i$, and the other endpoint is not a lower shortcut vertex. Then the shortest path does not enter the lower shortcut structure. Its bottleneck is either in $E_{\mathrm{small}}$ or is incident to the upper shortcut endpoint. For a fixed query vertex, this gives $O(P+Q)$ possible edges in $E_{\mathrm{mid}}$.

\item Suppose one endpoint is a lower shortcut vertex, namely $y'_{i,p}$ or $y''_i$. Then the bottleneck is either in $E_{\mathrm{small}}$ or is incident to a lower shortcut vertex. For a fixed query vertex, this gives $O(P+Q)$ possible edges in $E_{\mathrm{mid}}$.

\end{itemize}

\noindent\textbf{Case 2: at least one filler endpoint.}
We next consider the cases in which at least one endpoint is a filler vertex. 

Suppose first that both endpoints are filler vertices, say $f_{t,p,q}$ and $f_{t',p',q'}$.
\begin{itemize}
\item If $p=p'$ and $q=q'$, then for a fixed query vertex there are only $O(\tau)$ such targets, so these paths contribute at most $O(\tau)$ possible bottleneck edges.

\item If $p=p'$ and $q\neq q'$, then the shortest path uses the lower shortcut vertex $y'_{i,p}$ for the part $i$ chosen by the infinitesimal weight perturbations:
$$
 f_{t,p,q}-y_{i,p,q}-y'_{i,p}-y_{i,p,q'}-f_{t',p,q'}.
$$
This path adds only $O(\varepsilon_1)$ shortcut weight beyond the two necessary filler edges, whereas any path that changes addresses or uses the upper shortcut structure is strictly longer in base weight. 
Its bottleneck edge is one of the base-capacity-$3$ edges incident to $y'_{i,p}$. As $i$ and the copy index vary, there are only $O(kQ)=O(\tau)$ such edges.

\item If $p\neq p'$, then the shortest path changes addresses through $y''_i$ for the part $i$ chosen by the infinitesimal weight perturbations. The path has a base-capacity-$2$ edge incident to $y''_i$, and hence its bottleneck lies in $E_{\mathrm{small}}$.
\end{itemize}

Now suppose that one endpoint is the filler vertex $f_{t,p',q'}$, and the other endpoint is associated with part $i$. The shortest path in this case is similar to the case where both endpoints are in part $i$ and one of them is in the $y$-layer.
\begin{itemize}
\item If the other endpoint is a butterfly vertex $v_{i,\ell,p,q}$, then:
\begin{itemize}
    \item If $v_{i,h,p',q}\notin O_{i,\ell,p,q}$, then their shortest path goes through $y''_i$, and the bottleneck edge is incident to $y''_i$.
    \item Otherwise, the shortest path is
    $$
    v_{i,\ell,p,q}\rightsquigarrow v_{i,h,p',q}-x_{i,p',q}-y_{i,p',q'}-f_{t,p',q'}.
    $$
    \begin{itemize}
        \item[(a)] If $q\ne q_i^\star$, then the $v$-$x$ edge is non-high and has base capacity $4$, while all butterfly, $x$-$y$, and filler edges on this path have base capacity at least $5$. Hence the bottleneck is this $v$-$x$ edge.
        \item[(b)] If $q=q_i^\star$ and $v_{i,\ell,p,q}\in D^V_{i,p_i^\star,q_i^\star}$, then the butterfly and $x$-$y$ edges have base capacity $9$, the $v$-$x$ edge has base capacity $8$, and the filler edge has base capacity $7$. Hence the bottleneck is the filler edge $(y_{i,p',q'}, f_{t,p',q'})$.
        \item[(c)] Otherwise, the bottleneck edge is a deepest non-high butterfly edge on the butterfly suffix. Suppose it is $(v_{i,\ell'-1,p''',q_i^\star},v_{i,\ell',p'',q_i^\star})$. Then $\ell'\in\{1,\ldots,h\}$, and $v_{i,0,p_i^\star,q_i^\star}\in I_{i,\ell',p'',q_i^\star}$ and $v_{i,h,p',q_i^\star}\in O_{i,\ell',p'',q_i^\star}$.
    \end{itemize}

For a fixed query vertex $v_{i,\ell,p,q}$, subcases (a) and (c) reveal only $O(P\log P)$ edges in $E_{\mathrm{mid}}$. For subcase (b), the event $v_{i,\ell,p,q}\in D^V_{i,p_i^\star,q_i^\star}$ has probability $2^\ell/(PQ)$, and the number of addresses $p'$ with $v_{i,h,p',q}\in O_{i,\ell,p,q}$ is $2^{h-\ell}$. Together with the $\tau Q$ choices of $(t,q')$, the expected contribution is
$$
\frac{2^\ell}{PQ}\cdot 2^{h-\ell}\cdot \tau Q=\tau.
$$  

For a fixed filler query vertex $f_{t,p',q'}$, subcase (a) contributes $O(kQ)$ edges and subcase (b) contributes $O(k)$ edges. For subcase (c), fix a part $i$ and a level $\ell'\in\{1,\ldots,h\}$. There are $P/2^{\ell'}$ vertices $v_{i,\ell',p'',q_i^\star}$ such that $v_{i,h,p',q_i^\star}\in O_{i,\ell',p'',q_i^\star}.$ For each such vertex, the event $v_{i,0,p_i^\star,q_i^\star}\in I_{i,\ell',p'',q_i^\star}$ occurs with probability $2^{\ell'}/P$. The number of revealed edges is at most the sum of the numbers of butterfly edges entering these vertices, which is at most twice the number of such vertices. Thus each fixed part and level contributes at most
$$
2\cdot\frac{P}{2^{\ell'}}\cdot\frac{2^{\ell'}}{P}=O(1).
$$
Summing over all parts and levels gives $O(k\log P)=\tilde O(\tau)$.

\end{itemize}
\item If the other endpoint is an $x$-layer vertex $x_{i,p,q}$, then different addresses $p\neq p'$ are handled by the lower shortcut through $y''_i$, so the bottleneck lies in $E_{\mathrm{small}}$. If $p=p'$, the shortest path is
$$
 x_{i,p,q}-y_{i,p,q'}-f_{t,p,q'}.
$$
The bottleneck is the $x$-$y$ edge when $q\neq q_i^\star$, and is the filler edge when $q=q_i^\star$. For a fixed filler query this gives at most $O(kQ)=O(\tau)$ possible $E_{\mathrm{large}}$ bottlenecks over all parts and copies. For a fixed $x$-layer query, the $q\neq q_i^\star$ case gives at most $Q$ possible $x$-$y$ bottlenecks, while the event $q=q_i^\star$ has probability $1/Q$ and then gives at most $\tau Q$ filler-edge bottlenecks. The expected contribution is therefore $O(Q+\tau)=O(\tau)$.

\item If the other endpoint is a $y$-layer vertex $y_{i,p,q}$, then different addresses again lead to a bottleneck in $E_{\mathrm{small}}$. If $p=p'$ and $q=q'$, the path is the single filler edge $(y_{i,p,q},f_{t,p,q})$. If $p=p'$ and $q\neq q'$, the shortest path goes through $y'_{i,p}$ and its bottleneck is a base-capacity-$3$ edge incident to $y'_{i,p}$. Thus a fixed filler query sees $O(k)\le O(\tau)$ possible filler-edge bottlenecks and $O(kQ)=O(\tau)$ possible $E_{\mathrm{mid}}$ bottlenecks from $y$-layer targets. Conversely, a fixed $y$-layer query sees at most $O(\tau)$ filler-edge bottlenecks and $O(Q)$ possible $E_{\mathrm{mid}}$ bottlenecks.

\item If the other endpoint is an upper shortcut vertex, then $v''_i$ contributes only $E_{\mathrm{small}}$ bottlenecks because every path from $v''_i$ uses an incident base-capacity-$2$ shortcut edge. For an endpoint $v'_{i,q}$, the shortest path to $f_{t,p',q'}$ goes through the bottom-level vertex $v_{i,h,p',q}$ and then through $x_{i,p',q}$ and $y_{i,p',q'}$; its bottleneck is the base-capacity-$3$ edge $(v_{i,h,p',q},v'_{i,q})$. Hence a fixed filler query obtains $O(kQ)=O(\tau)$ possible $E_{\mathrm{mid}}$ bottlenecks from such targets, while a fixed query at $v'_{i,q}$ obtains at most $O(P)$ such bottlenecks as the filler address varies.

\item If the other endpoint is a lower shortcut vertex $y'_{i,p}$ or $y''_i$, then addresses $p\neq p'$ and the vertex $y''_i$ give bottlenecks in $E_{\mathrm{small}}$. For $y'_{i,p'}$, the bottleneck is one of the base-capacity-$3$ edges incident to $y'_{i,p'}$, giving $O(k)$ possible bottlenecks for a fixed filler query and $O(Q)$ for a fixed lower-shortcut query.

\end{itemize}

\noindent\textbf{Case 3: endpoints in different parts, or one endpoint is $y'''$.}
Let the two endpoints be non-filler vertices associated with different parts. The shortest path uses the shared lower shortcut structure and passes through $y'''$: any path through a filler vertex would use at least two filler edges of base weight $1$, whereas using the lower shortcut adds only $O(\varepsilon_1)$ shortcut weight once the path reaches the appropriate $y$-layers. Therefore the bottleneck is the base-capacity-$1$ edge incident to $y'''$, which lies in $E_{\mathrm{small}}$. The same conclusion holds when one endpoint is $y'''$, including the case where the other endpoint is a filler vertex.

The case analysis above gives, for every fixed query vertex $u$,
$$
\mathbb{E}_{\sigma}[|R(u)\cap E_{\mathrm{large}}|]=O(\tau)
\qquad\text{and}\qquad
\mathbb{E}_{\sigma}[|R(u)\cap E_{\mathrm{mid}}|]=O(\tau+P\log P)=\tilde O(\tau+P).
$$
This proves \Cref{lem:lower-bound-single-query}.
\end{proof}

We now extend the fixed-query-set upper bound to randomized non-adaptive algorithms. Fix an arbitrary randomized non-adaptive algorithm $\mathcal A$ that makes at most $\alpha k$ queries. Because $\mathcal A$ is non-adaptive, its random query set is independent of $\sigma$. Equivalently, we may decompose $\mathcal A$ into deterministic non-adaptive algorithms, that is, fixed query sets $S_1,\ldots,S_m$ with $|S_j|\le \alpha k$, chosen with probabilities $p_1,\ldots,p_m$. Then, by the fixed-query-set upper bound,
\begin{align*}
\mathbb{E}_{\sigma,\mathcal A}[\ALG_{\alpha k}]
&\le \sum_j p_j\mathbb{E}_{\sigma}[|R(S_j)|]\\
&\le \sum_j p_j\,\tilde{O}(\alpha k\tau)\\
&= \tilde{O}(\alpha k\tau).
\end{align*}
Therefore, by averaging, there exists a fixed vector $\sigma^\star$ such that on the corresponding fixed capacity assignment,
$$
\mathbb{E}_{\mathcal A}[\ALG_{\alpha k}\mid \sigma=\sigma^\star]
\le \tilde{O}(\alpha k\tau).
$$

\paragraph{Lower bound on $\mathrm{OPT}_k$ and conclusion.}
Fix the vector $\sigma^\star=(\sigma_i^\star)_{i\in\Zx{k}}$ obtained above, and write $\sigma_i^\star=(p_i^\star,q_i^\star)$. We show that $\OPT_k$ is large by considering the following $k$ queries:
$$
v_{i,0,p_i^\star,q_i^\star},
\qquad i\in\Zx{k}.
$$
Fix a part $i$. For every $t\in\Zx{\tau}$, $p\in\Zx{P}$, and $q\in\Zx{Q}$, the shortest path from $v_{i,0,p_i^\star,q_i^\star}$ to $f_{t,p,q}$ is the path through part $i$:
$$
v_{i,0,p_i^\star,q_i^\star}
\rightsquigarrow
v_{i,h,p,q_i^\star}
-x_{i,p,q_i^\star}-y_{i,p,q}-f_{t,p,q}.
$$
The butterfly edges and the $x$-$y$ edge on this path are high edges with base capacity $9$, and the $v$-$x$ edge is a high edge with base capacity $8$. The final filler edge $(y_{i,p,q},f_{t,p,q})$ has base capacity $7$. Since the infinitesimal perturbations preserve all comparisons between different base capacities, the bottleneck edge on this path is
$$
b(v_{i,0,p_i^\star,q_i^\star},f_{t,p,q})=(y_{i,p,q},f_{t,p,q}).
$$
Thus, the query $v_{i,0,p_i^\star,q_i^\star}$ reveals all $PQ\tau$ filler edges incident to the $y$-layer of part $i$. As $i$ ranges over $\Zx{k}$, these sets of filler edges are edge-disjoint, because the endpoint in the $y$-layer carries the part index $i$, even though the filler vertices themselves are shared across parts. Hence
$$
\OPT_k
\ge kPQ\tau.
$$
By the choice of $P$ in the size-and-padding step,
$$
\frac{N_0}{4}\le PQ\tau\le N_0,
$$
and since $N_0=\lfloor n/\Lambda\rfloor$ with $\Lambda=\Theta(\log^2 n)$, we have $PQ\tau=\tilde{\Theta}(n)$. Therefore
$$
\OPT_k\ge \tilde{\Omega}(kn).
$$
Combining this with the upper bound for $\mathcal A$ on the capacity assignment corresponding to $\sigma^\star$ gives
$$
\frac{\OPT_k}{\mathbb{E}_{\mathcal A}[\ALG_{\alpha k}\mid \sigma=\sigma^\star]}
\ge
\tilde{\Omega}\left(\frac{n}{\alpha\tau}\right).
$$
Finally, by the definition of $\tau$,
$$
\frac{n}{\alpha\tau}
=
\tilde{\Theta}\left(\min\left\{\frac{n}{\alpha k},\max\left\{\sqrt{\frac n\alpha},\frac{n^{2/3}}{\alpha k^{1/3}}\right\}\right\}\right).
$$
This proves the lower bound part of \Cref{thm:main}.

\paragraph{AI disclosure.} We used ChatGPT 5.5 Thinking to assist with verifying the case-by-case analysis in \Cref{sec:lb}. The tool helped identify several flaws in earlier drafts of the argument. The substantive argument in the final write-up was developed by the authors.


\bibliographystyle{plainurl}
\bibliography{bottleneck}

\end{document}